\documentclass[12pt]{article} 

\usepackage{url,fullpage,color}
\usepackage{graphicx}
\usepackage{amsmath, amsthm, amssymb}
\usepackage{algpseudocode}
\usepackage{algorithm}
\usepackage{subfigure}
\usepackage{comment}
\RequirePackage[round]{natbib}
\RequirePackage[colorlinks,citecolor=blue,urlcolor=blue]{hyperref}

\usepackage{setspace}

\usepackage{multirow}

\newtheorem{thm}{Theorem}
\newtheorem{lem}[thm]{Lemma}

\let\hat\widehat

\theoremstyle{remark}
\newtheorem{remark}{Remark}

\newcommand\R{\mathbb{R}}
\newcommand\K{\mathbb{K}}

\renewcommand\P{\mathbb{P}}
\newcommand\mathand{\ {\rm and}\ }
\newcommand\norm[1]{\|#1\|}
\newcommand\Haus{{\sf Haus}}
\newcommand\dest{{\sf dest}}

\newcommand\cL{{\cal L}}
\newcommand\cF{{\cal F}}
\newcommand\cM{{\cal M}}
\newcommand\cC{{\cal C}}
\newcommand\cA{{\cal A}}
\newcommand\cB{{\cal B}}

\DeclareMathOperator{\Cov}{Cov}

\DeclareMathOperator{\Quantile}{{\sf Quantile}}
\DeclareMathOperator{\reach}{{\sf reach}}
\newcommand\BC{\mathbf{BC}}

\newcommand{\Set}[1]{\left\{{#1}\right\}}       
\newcommand{\st}{:\;}

\catcode`@=11
\newskip\beforeproofvskip
\newskip\afterproofvskip
\beforeproofvskip=\medskipamount
\afterproofvskip=\bigskipamount

\def\prooftag{Proof}
\def\proofskip{\enspace}

\def\proof{\@ifnextchar[{\@@proof}{\@proof}}  
\def\@startproof{\par\vskip\beforeproofvskip\leavevmode}
\def\@proof{\@startproof{\scshape\prooftag.}\proofskip}
\def\@@proof[#1]{\@startproof {\scshape\prooftag #1.}\proofskip}

\catcode`@=12

\newcommand{\blind}{0}

\begin{document}

\def\spacingset#1{\renewcommand{\baselinestretch}%
{#1}\small\normalsize} \spacingset{1}

\if0\blind
{
  \title{\bf Density Level Sets: Asymptotics, Inference, and Visualization}
  \author{ Yen-Chi Chen\thanks{Supported by the William S. Dietrich II Presidential Ph.D. Fellowship and DOE grant number DE-FOA-0000918 at Carnegie Mellon University.
	}\hspace{.2cm}\\
    Department of Statistics, University of Washington\\
    and \\
    Christopher R. Genovese\thanks{Supported by DOE grant number DE-FOA-0000918 and 
    NSF grant number DMS-1208354.
	} \\
    Department of Statistics, Carnegie Mellon University\\
    and \\
    Larry Wasserman\thanks{Supported by NSF grant number DMS-1208354.} \\
    Department of Statistics, Carnegie Mellon University}
  \maketitle
} \fi

\if1\blind
{
  \bigskip
  \bigskip
  \bigskip
  \begin{center}
    {\LARGE\bf 
    	\vspace{1 in}
    Density Level Sets: Asymptotics, Inference, and Visualization}
\end{center}
  \medskip
} \fi

\bigskip

\begin{abstract}
\noindent
We study the plug-in estimator for density level sets under Hausdorff
loss. We derive asymptotic theory for this estimator, and based on this
theory, we develop two bootstrap confidence regions for level sets. We
introduce a new technique for visualizing density level sets, even in
multidimensions, that is easy to interpret and efficient to compute.
\end{abstract}

\noindent%
\emph{Keywords:}  
Nonparametric inference, asymptotic theory, level set clustering, anomaly detection, visualization
\vfill

\newpage
\spacingset{1.45} 
\section{Introduction}

Estimating the level sets of a probability density function
has a wide range of applications, 
including anomaly detection
(outlier detection) \citep{breunig2000lof, hodge2004survey},
two-sample comparison \citep{duong2009highest},
binary classification \citep{mammen1999smooth},
and clustering \citep{Rinaldo2010a,Rinaldo2010b}.
In this paper,
we study the problem of estimating
the level set 
\begin{equation}
D_h \equiv D_h(\lambda) = \Set{x\st p_h(x)=\lambda},
\end{equation}
where $p_h$ is the expected kernel density estimator with bandwidth $h$,
a smoothed version of the underlying density $p$.
Using $D_h$ (and thus $p_h$) has several advantages,
which we discuss in detail in Section \ref{sec::smooth}.
Figures \ref{fig::vis_two} and \ref{fig::vis1}
illustrate the kind of confidence sets and visualizations
we will develop in this paper.

A commonly used estimator of the density level-set is the plug-in estimator
$\hat{D}_h = \Set{x\st \hat{p}_h(x) =\lambda}$,
where $\hat{p}_h$ is the kernel density estimator or some other density estimator.  
There is a large literature for level sets 
(and upper level sets, which replace $= \lambda$ with $\ge\lambda$) 
that focuses on 
the consistency, rates of convergence
\citep{Polonik1995, Tsybakov1997, Walther1997,Cadre2006, Cuevas2006}
and minimaxity \citep{singh2009adaptive}
of such estimators
under various error loss functions.

Recent results on statistical inference for level sets
include \cite{jankowski2012confidence} and \cite{mammen2013confidence}.
Statistical inference is challenging in this setting 
because the estimand is a set
and the estimator is a random set \citep{Molchanov2005}.
\cite{mason2009asymptotic} establish
asymptotic normality for upper level sets when the loss function
is the measure of the set difference.
However, it is unclear how to derive a confidence set from this result.

Another challenge of level set estimation
is that we cannot directly visualize
the level sets when the dimension of the data $d$ is larger than 3.
One approach  is to construct a \emph{level-set tree},
which shows how the connected components 
for the upper level sets 
bifurcate
when we gradually increase $\lambda$
\citep{stuetzle2003estimating, klemela2004visualization, klemela2006visualization,
klemela2009smoothing, stuetzle2010generalized,kent2013debacl}.
The level-set tree reveals topological information
about the level sets
but loses geometric information.

In this paper, we propose solutions to all of these problems.
Our main contributions can be summarized as follows.
\begin{enumerate}
\item We derive the limiting distribution of 
$\Haus(\hat{D}_h, D_h)$  
(Theorem~\ref{thm::Gaussian}).
\item We develop two bootstrap-based methods to construct confidence regions for $D_h$ 
(Section \ref{sec::inf}).
\item We prove that both bootstrap methods are valid
(Theorem \ref{thm::bootstrap} and \ref{thm::CI2}).
\item We devise a visualization technique that preserves the geometric information for density level sets
(Section \ref{sec::vis}).
\end{enumerate}

\emph{Related Work.}
Early work on density level set focuses on proving the consistency or
the rate of convergence under various metrics. See e.g.
\cite{Polonik1995, Tsybakov1997, Walther1997, Cuevas2006,Rinaldo2010a}.
However, none of these derives a limiting distribution for 
the density level sets.
To our knowledge, the only paper that considers limiting distributions is
\cite{mason2009asymptotic},
proving asymptotic normality under a generalized integrated distance.
However, this metric cannot be used to construct a confidence 
set for density level sets since the asymptotic distribution
involves the true density, which is unknown. 
Estimating the level set is also related to support estimation,
see e.g. \cite{cuevas2004boundary} and \cite{Cuevas2009}.

\cite{jankowski2012confidence} and \cite{mammen2013confidence},
both provide methods for constructing confidence sets for
the density level sets using the variation of the density function.
Our approach is similar to theirs but is based on Hausdorff distance.
We will compare their methods to ours in Section \ref{sec::CI::M2}.

\emph{Outline.}
We begin with a short introduction to density level sets 
along with some useful geometric concepts in Section \ref{sec::tech}.
In Section \ref{sec::thm}, we derive the limiting distribution of
the Hausdorff distance between the estimated and true level sets.
In Section \ref{sec::inf}, 
we construct a valid confidence set for density level sets.
In Section \ref{sec::vis}, we devise a visualization method for density level sets 
that is simple to interpret and efficient to compute.
(We provide an R package that implements our visualization method.)
We summarize our results and discuss related problems
in Section \ref{sec::discuss}.

\section{Technical Background} \label{sec::tech}

\subsection{Level Sets}

Let $X_1,\cdots,X_n$ be a random sample from an unknown, continuous density $p(x)$.
We define the density level set by
\begin{equation} \label{eq::p-level-set}
D \equiv D(\lambda) = \Set{x\st p(x)= \lambda}
\end{equation}
for some $\lambda>0$. 
Note that in the literature the term level sets is sometimes
used for the set $\Set{x\st p(x)\geq \lambda}$;
we call the latter the \emph{upper level set} to distinguish
it from the \emph{level set} in equation \ref{eq::p-level-set}.
Thus, the level sets (in our terminology)
are the boundaries to the upper level sets,
under mild smoothness assumptions (e.g., assumption G below).

We assume that $\lambda$ is a fixed, positive value.
A plug-in estimate for $D$ is
\begin{equation}
\hat{D}_h \equiv \hat{D}_h(\lambda) = \{x: \hat{p}_h(x)= \lambda\},
\end{equation}
where $\hat{p}_h$ is the kernel density estimator (KDE),
\begin{equation}
\hat{p}_h(x) = \frac{1}{nh^d}\sum_{i=1}^n K\left(\frac{\norm{x-X_i}}{h}\right).
\end{equation}

\subsection{Smoothed Density Level Set}	\label{sec::smooth}

In this paper, 
we focus on inference for the level sets of a smoothed version of $p$, 
specifically:
\begin{equation}
p_h = p\star K_h = \mathbb{E}(\hat{p}_h),
\end{equation}
where $K_h(x) = \frac{1}{h^d}K\left(\frac{\norm{x}}{h}\right)$ 
and $\star$ denotes convolution.
We denote the $\lambda$ level set of $p_h$ by
\begin{equation}
D_h = \Set{x\st p_h(x)= \lambda}.
\end{equation}
Note that although we focus on estimating $D_h$,
we allow $h=h_n\rightarrow 0$ as $n\rightarrow \infty$.

Here, we argue that 
$p_h$ (and hence $D_h$) is a better target for level-set estimation than $p$ (and hence $D$).
For simplicity, in this section we focus on 
the upper level sets
$L_h(\lambda) = \{x:\ p_h(x) \geq \lambda\}$
and $L(\lambda) = \{x: p_h(x)\geq\lambda\}$,
but the gist of the argument remains the same in either case.

Our arguments are:
(i)~$p_h$ always exists while $p$ may not even exist;
(ii)~$\hat p_h$ can be naturally viewed as an estimator for $p_h$;
(iii)~$p_h$ has all the salient structure of $p$ but is easier to estimate than $p$;
(iv)~While the bias $p_h(x) - p(x)$ may be analyzed theoretically,
in practice, it cannot be accurately estimated.
It is better to just be clear that we are really estimating $p_h$.

Let us now expand on some of these points.
Regarding point (ii):
Let $X_1,\ldots, X_n$ be an iid sample from $P$,
where $P$ is supported on some compact set
$\K \subset \mathbb{R}^d$.
We have, for any $P$,
$$
P^n\Biggl(\sup_x |\hat p_h(x) - p_h(x)| \geq \epsilon\Biggr) \leq c_1 e^{- c_2 n h^d \epsilon^2 }.
$$
As long as $h\geq (\log n/n)^{1/d}$,
$P^n(\sup_x |\hat p_h(x) - p_h(x)| \geq \epsilon) \to 0$,
and hence we can uniformly consistently estimate $p_h$,
whether we keep $h$ fixed or let it tend to 0.
The same is \emph{not} true for $p$.
In fact, $P$ may not even have a density $p$.
The bias cannot be uniformly estimated
in a distribution-free way.

Regarding (iii):
The left plot in Figure~\ref{fig::j1} shows a density $p$.
The blue points at the bottom show the upper level set $L=\{x:\ p \geq 0.05\}$.
The right plot shows
$p_h$ for $h=0.2$ and
the blue 
points at the bottom show the upper level set $L_h=\{x:\ p_h \geq 0.05\}$.
The smoothed out density $p_h$ is biased
and the upper level set $L_h$ loses the small details of $L$.
But these small details are the least estimable,
and $L_h$ captures the principal structure of $L$.
In addition, when $L$ is smooth (having positive $\mu$-reach; \citealt{Chazal2005,Chazal2009}) 
and $h$ is sufficiently small, 
$L_h$ and $L$ will be topologically similar, in the sense that a small expansion of $L_h$
is homotopic to $L$
(\citealt{Chazal2005,Chazal2009, genovese2014nonparametric}).
The $\mu$-reach and the concept of being nearly homotopic can be found
in section 3.2 of \cite{genovese2014nonparametric} and section 4 of \cite{Chazal2009}.

As a second example,
let $P= (1/3) \phi(x;-5,1) + (1/3) \delta_0 + (1/3) \phi(x;5,1)$
where $\phi$ is a Normal density and $\delta_0$ is a point mass at 0.
Of course, this distribution does not even have a density.
The left plot in Figure
(\ref{fig::mix}) shows the density of the absolutely continuous
part of $P$ with a vertical line to who the point mass.
The right plot shows $p_h$, which is a smooth, well-defined density.
Again the blue points show the level sets.
As before $p_h$ and $L_h$ are slightly biased,
but they are also well-defined.
And $\hat p_h$ and $\hat L_h$
are accurate estimators of $p_h$ and $L_h$, respectively.
Moreover, $L_h$ captures the most important 
qualitative information about $L$, namely, that there are
three connected components, one of which is small.
These examples show that $p_h$ --- and hence $D_h$ ---
is a sensible target of inference.

\begin{figure}
\begin{center}
\begin{tabular}{cc}
\includegraphics[scale=.4]{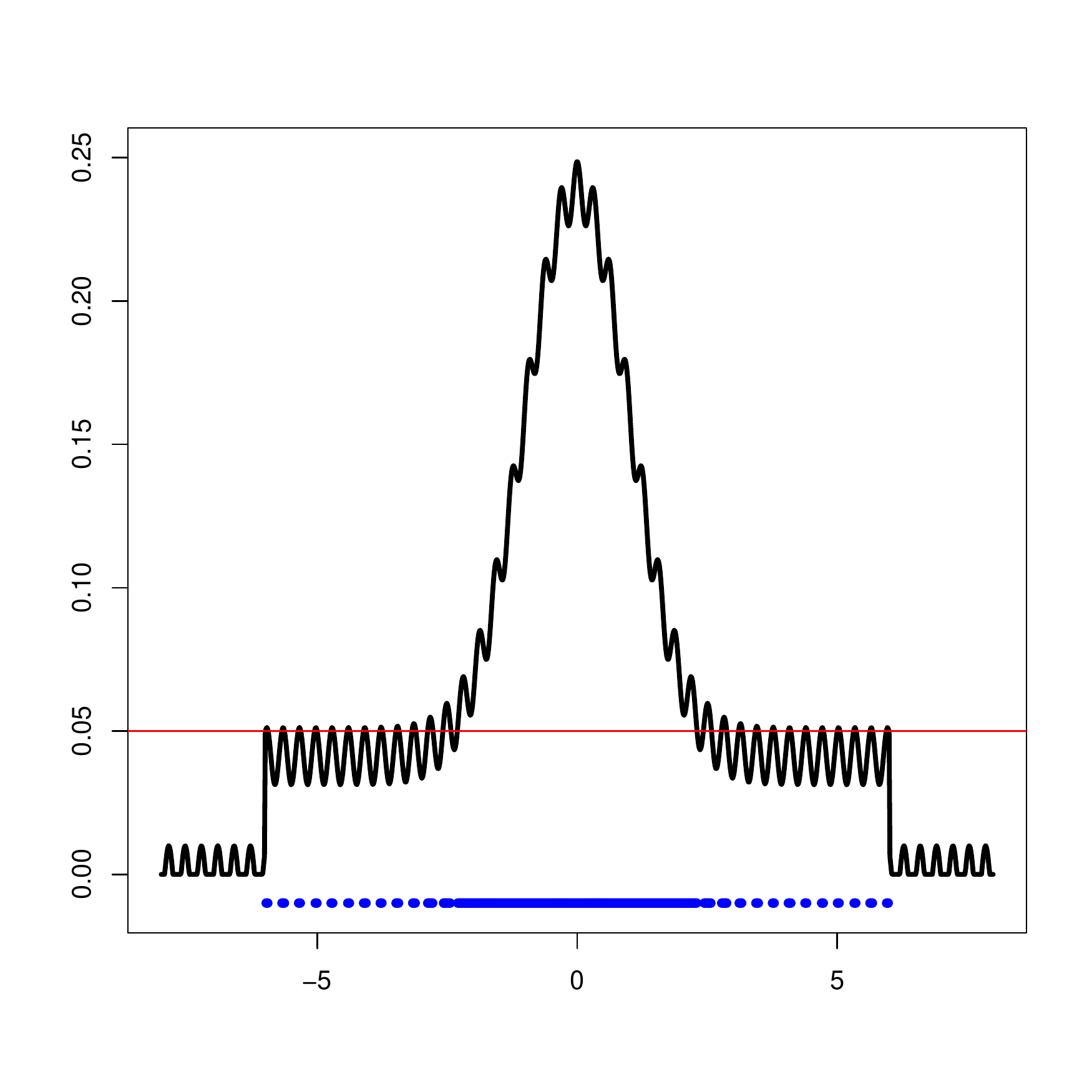} & \includegraphics[scale=.4]{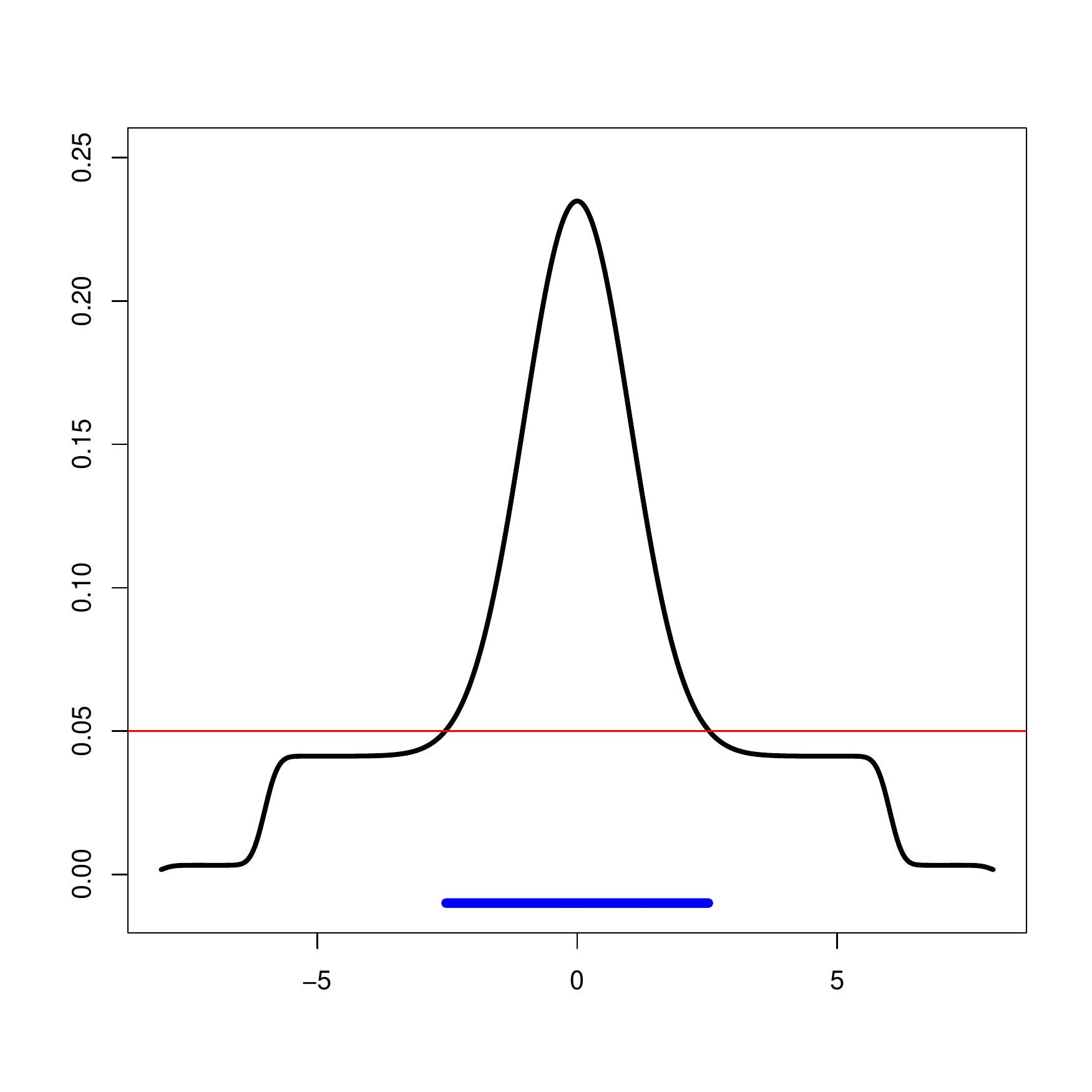}
\end{tabular}
\end{center}
\caption{Left: a density $p$ and an upper level set $\{p \geq t\}$.
Right: the smoothed density $p_h$ and the upper level set $\{p_h \geq t\}$.}
\label{fig::j1}
\end{figure}

\begin{figure}
\begin{center}
\begin{tabular}{cc}
\includegraphics[scale=.4]{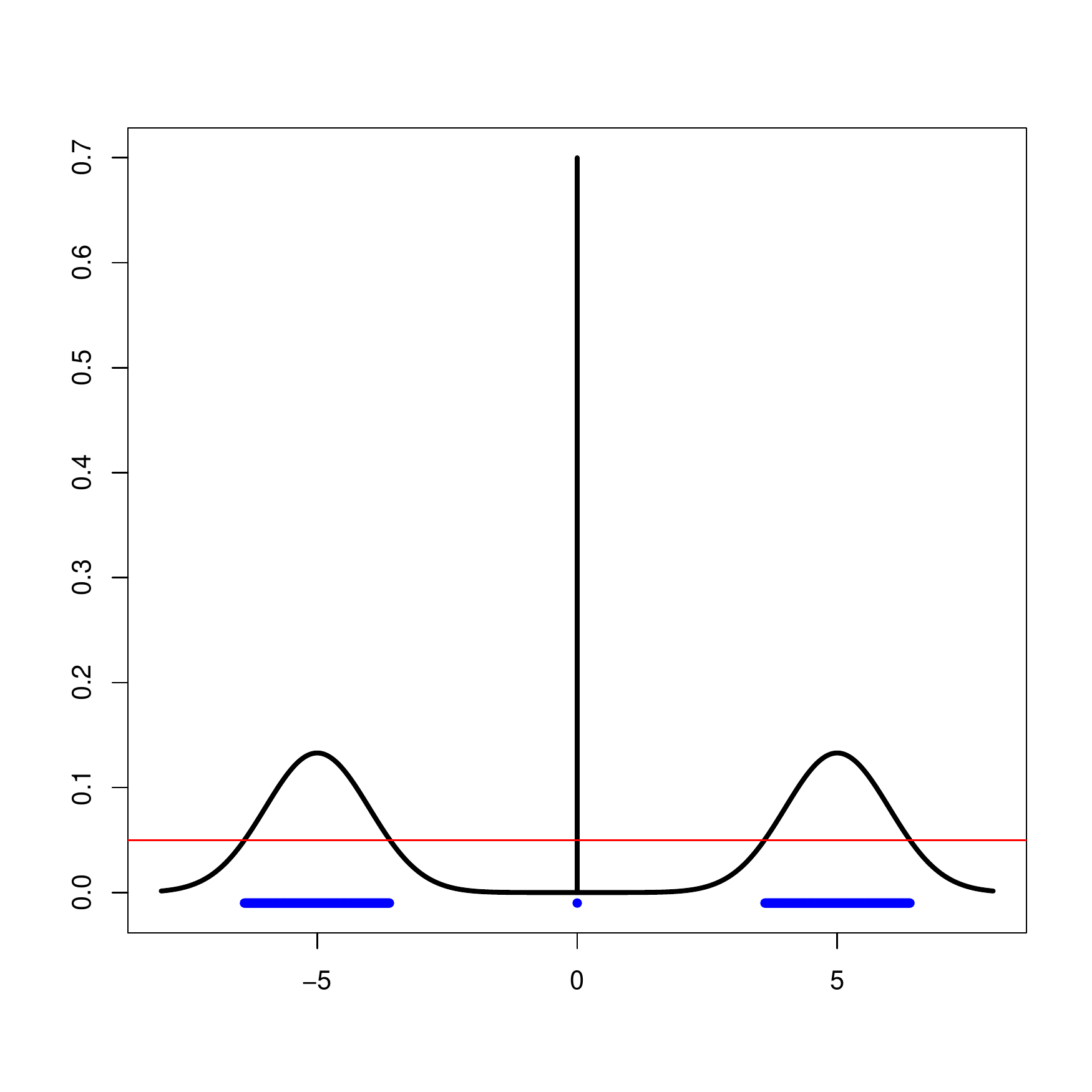} & \includegraphics[scale=.4]{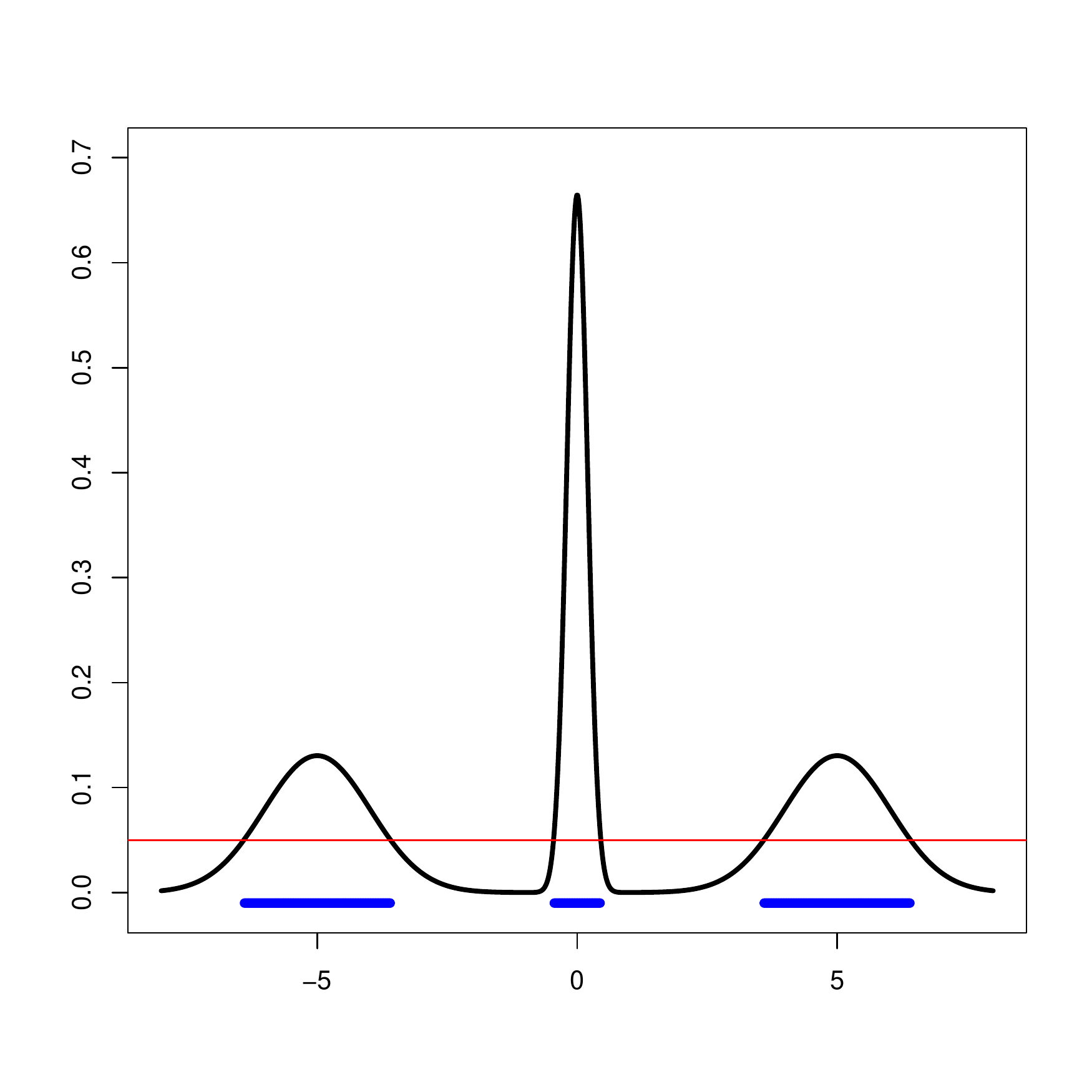}
\end{tabular}
\end{center}
\caption{Left: a distribution with a continuous component and a point mass at 0.
Right: the smoothed density $p_h$. The upper level set $L_h$ is biased but is estimable
and it approximates the main features of $L$.}
\label{fig::mix}
\end{figure}

Lastly, we would like to point out that
the idea of viewing $p_h$
as the estimand is not new.
The ``scale space'' approach to smoothing
explicitly argues that we should view
$\hat p_h$ as an estimate of $p_h$,
and $p_h$ is then regarded as a view of $p$ at a particular resolution.
This idea is discussed in detail in
\cite{chaudhuri2000scale, chaudhuri1999sizer, godtliebsen2002significance}.

If one really wants to focus on making inference for $D$, the level
set of the original density, then 
we need to undersmooth\footnote{Undersmoothing the density estimate to make statistical inferences 
is a common practice in nonparametric statistics; see e.g. page 89 of
\cite{wasserman2006all}.} so that the bias will not affect the limiting distribution.
This leads to estimates of $D$ that are highly variable.
We believe that an accurate confidence set for $D_h$
is more useful than a poor confidence set for $D$.

\begin{remark}
These arguments explain why we regard $p_h$ rather than $p$
as the estimand.
But these arguments do not tell us how to choose $h$.
Bandwidth selection is always a challenge and in this paper 
we mainly use Silverman's rule of thumb
(\cite{Silverman1986}).
\end{remark}

\subsection{Geometric Concepts}

Let $\pi_A(x)$ be the \emph{projection} of a point $x$ onto a set $A$. 
Note that $\pi_A(x)$ 
may not be unique.
The distance induced by the projection is 
\begin{equation}
d(x,A) = \inf\{\norm{x-y}_2: y\in A\} = \norm{x-\pi_A(x)}_2 
\end{equation}

A common measure of distance between two subsets of a metric space (e.g., $\R^d$)
is the \emph{Hausdorff distance}, given by
\begin{equation}
\begin{aligned}
\Haus(A,B) &= \inf\{\epsilon: A \subset B \oplus \epsilon \mathand B \subset A \oplus \epsilon\}\\
& = \max\left\{\sup_{x\in B}d(x, A) , \sup_{x\in A}d(x, B)\right\},
\end{aligned}
\label{eq::Haus}
\end{equation}
where $A \oplus \epsilon = \bigcup_{x\in A} B(x,\epsilon)$ and 
$B(x,\epsilon)=\{y:\; \norm{x-y} \le \epsilon\}$.
The Hausdorff distance is a generalized version of the $\cL_{\infty}$ metric
for sets.

The \emph{reach} of a set $M$ (\citealt{Federer1959,Cuevas2009}, also known as condition
number~\citealt{Niyogi2008} or minimal feature size~\citealt{Chazal2005})
is the largest distance from $M$ such that
every point within this distance to $M$ has a unique projection onto
$M$. i.e.
\begin{equation}
\reach(M) = \sup\{r: \pi_M(x) \mbox{ is unique } \forall x\in M\oplus r\}.
\end{equation}
Note that $\pi_A(x)$ is unique if $0<d(x,A)\leq \reach(A)$.
Another way to understand the reach is as the largest radius of a ball that can 
freely move along $M$; see Figure~\ref{Fig::reach} for an example. In some cases, 
the reach is the same as the smallest radius of curvature on $M$. 
The reach plays a key role in relating the Hausdorff distance to the empirical process. 
Note that the reach is closely related to `rolling properties' 
and `$\alpha$-convexity'; see \cite{Cuevas2009}, \cite{cuevas2012statistical} 
and appendix A of \cite{Lopez2008}.

\begin{figure}
\centering
	\subfigure[]
	{
		\includegraphics[scale=0.5]{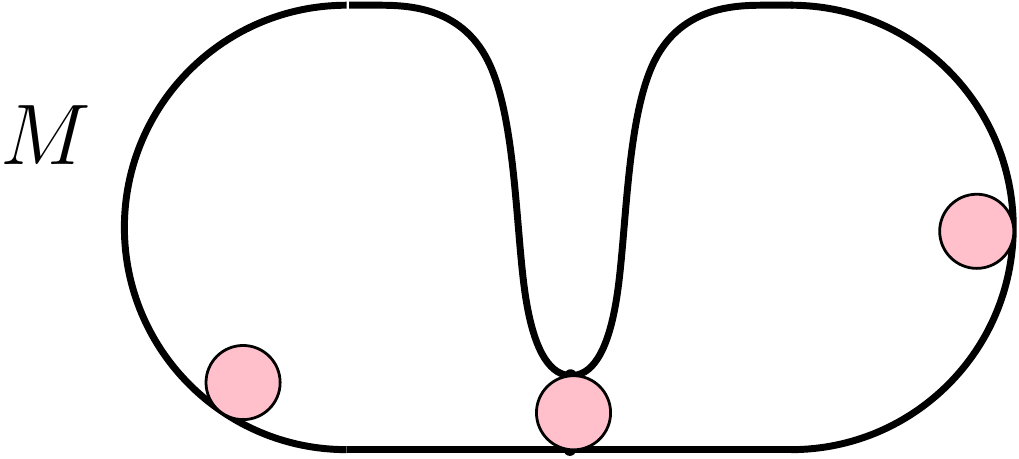}
	}
	\subfigure[]
	{
		\includegraphics[scale=0.5]{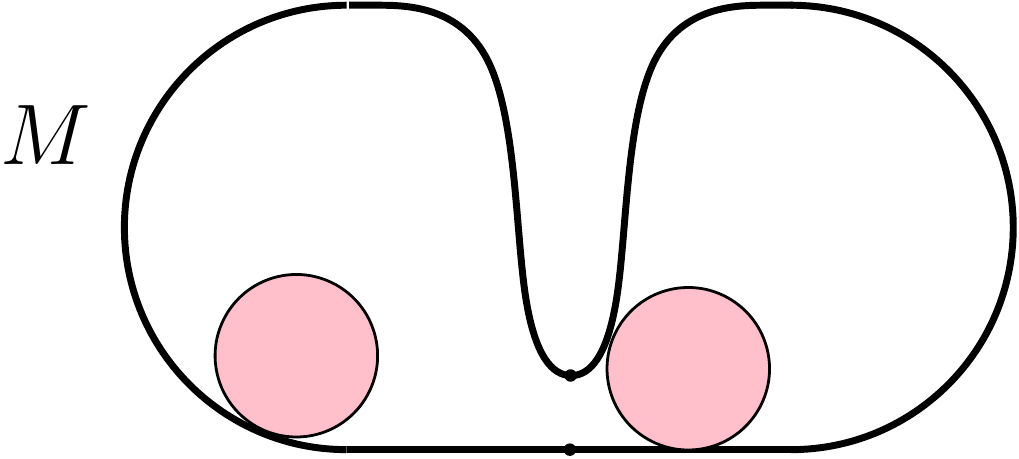} 
	}
\caption{An illustration for reach. The reach is the 
largest radius for a ball that can freely move along the 
set $M$.
In (a), the radius of the pink ball is equal to the reach. 
In (b), the radius is too large so that it cannot pass the small gap on $M$.}
\label{Fig::reach}
\end{figure}

Finally, two smooth sets $A$ and $B$ are called \emph{normal compatible} \citep{Chazal2007}
if the projection between $A$ and $B$ are one to one and onto; see
Figure~\ref{Fig::normal} for an example.
When $A$ and $B$ are normal compatible, the Hausdorff distance 
has the simpler form
\begin{equation}
\Haus(A,B) = \sup_{x\in B}d(x, A) = \sup_{x\in A}d(x, B).
\end{equation}

\begin{figure}
\centering
	\subfigure[]
	{
		\includegraphics[height= 1.5 in]{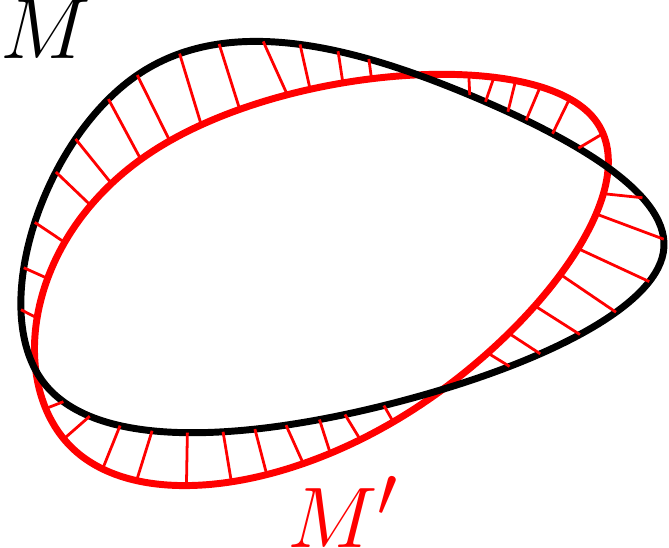}
	}
	\subfigure[]
	{
		\includegraphics[height=1.5 in]{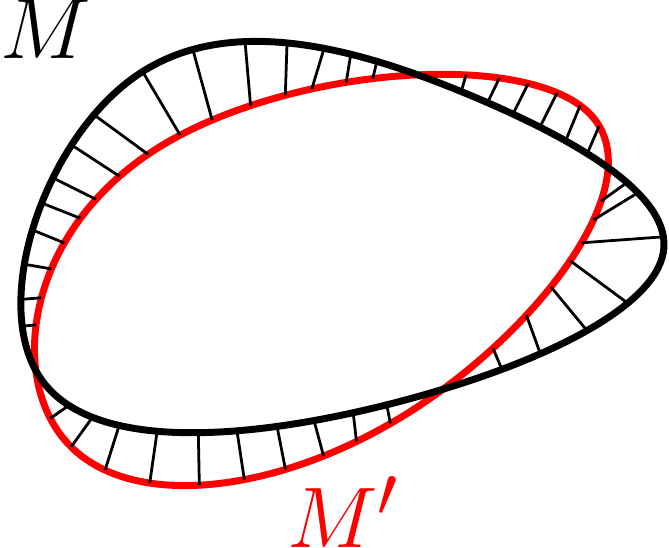} 
	}
\caption{
An example for two normal compatible curves.
Panel (a): each thin red line indicates the projection from a point of $M$ onto $M'$.
Panel (b): each thin black line indicates the projection from a point of $M'$ onto $M$.
When these projections are one to one and onto, we say $M$ and $M'$
are normal compatible to each other.
}
\label{Fig::normal}
\end{figure}

\section{Asymptotic Theory}
\label{sec::thm}

In this section, we derive the asymptotic theory for $\Haus (\hat{D}_h,D_h)$.
Note first that for two sets $A$ and $B$,
the Hausdorff distance satisfies the inclusion property
\begin{equation}
A\subset B\oplus\Haus(A,B), \,\, B\subset A\oplus \Haus(A,B).
\end{equation}
From this, it follows that
when we have a set estimator $\hat{A}_n$ and a parameter of interest
$A$, then for any $\alpha>0$, 
the set 
\begin{equation}
\hat{A}_n \oplus {\sf Quantile}\left(\Haus(\hat{A}_n,A), 1-\alpha\right)
\end{equation}
is a $1-\alpha$ confidence set for $A$,
where $\Quantile(X,\alpha)$ is the $\alpha$-quantile of random variable $X$.
Thus, whenever we can approximate the distribution of $\Haus(\hat{A}_n,A)$,
we can construct a confidence set for $A$.
Note that 
\cite{chen2014asymptotic} and \cite{chen2014nonparametric}
have used this property to construct confidence sets,
but neither paper used this property to full effect.

We now define some notation that will be used throughout this paper.
Let $\BC^r$
denote the collection of functions (including both univariate and multivariate functions) 
with bounded continuous derivatives up
to the $r$-th order.
For a smooth function $f:\R^d\mapsto \R$ with $f\in\BC^r$, 
we denote the elementwise max norm
for $r$-th derivative as $\norm{f}_{r,\max}$.
For instance, 
\begin{align*}
\norm{f(x)}_{1,\max}  = \max_{1\leq i\leq d}\left|\frac{\partial f(x)}{\partial x_i}\right|, \quad
\norm{f(x)}_{2,\max}  = \max_{1\leq i, j\leq d}\left|\frac{\partial^2 f(x)}{\partial x_i\partial x_j}\right|.
\end{align*}
We define the sup norm using derivatives up to $r$-th order by:
\begin{equation}
\norm{f}^*_{r,\max} = \max\left\{\sup_{x\in\K}\norm{f(x)}_{\ell,\max}: \ell= 0,\cdots, r\right\}.
\end{equation}
Let $\alpha = (\alpha_1,\cdots,\alpha_d)$ be an multi-index
such that each $\alpha_j$ is a non-negative integer
and $|\alpha| = \alpha_1+\cdots+\alpha_d$.
We define 
$$
f^{(\alpha)} = \frac{\partial^{|\alpha|} }{\partial^{\alpha_1} x_1\cdots\partial^{\alpha_d} x_d}f(x)
$$
to be the partial derivative.
Note that for a vector $v\in\R^d$, the norm $\norm{v}$ denotes the usual Euclidean norm.
Throughout this paper, we use the conventional notation for $O_P(\cdot)$ and $O(\cdot)$
for $n\rightarrow \infty$ and $h=h_n$, where possibly $h_n\rightarrow0$. 

\medskip

We now state are main assumptions,
for an arbitrary density $q$.
When we use the assumptions in what follows,
we will take $q$ to be $p_h$.

\noindent{\bf Assumptions.}
\begin{description}
\item[\textbf{(G)}] Let $D(q) =\{x\in \K: q(x)=\lambda\}$ be the level set for a density $q$. There are $\delta_0,g_0>0$ such that $\forall x\in D(q)\oplus \delta_0$, we have $\norm{\nabla q(x)}>g_0$.
\item[\textbf{(K1)}] The kernel function $K\in\BC^3$ and is symmetric, non-negative, and 
$$\int x^2K^{(\alpha)}(x)dx<\infty,\qquad \int \left(K^{(\alpha)}(x)\right)^2dx<\infty
$$ 
for all multi-indices $\alpha$ satisfying $|\alpha|\leq 3$. 
\item[\textbf{(K2)}] The kernel function $K$ and its partial derivative satisfies condition $K_1$ in \cite{Gine2002}. Specifically, let 
\begin{equation}
\begin{aligned}
\mathcal{K} &= \Set{y\mapsto K^{(\alpha)}\left(\frac{x-y}{h}\right)\st x\in\mathbb{R}^d, h>0, |\alpha|\leq2}\\
\end{aligned}
\end{equation}
We require that $\mathcal{K}$ satisfies
\begin{align}
\underset{P}{\sup} N\left(\mathcal{K}, L_2(P), \epsilon\norm{F}_{L_2(P)}\right)\leq \left(\frac{A}{\epsilon}\right)^v
\label{eq::VC}
\end{align}
for some positive numbers $A$ and $v$,
where $N(T,d,\epsilon)$ denotes the $\epsilon$-covering number of the metric space $(T,d)$, 
$F$ is the envelope function of $\mathcal{K}$, 
and the supremum is taken over the whole $\mathbb{R}^d$.
The $A$ and $v$ are usually called the VC characteristics of $\mathcal{K}$.
The norm $\norm{F}^2_{L_2(P)} =\int|F(x)|^2dP(x)$.
\end{description}
Assumption (G) appears in 
\cite{Molchanov1990,Tsybakov1997,Walther1997, Molchanov1998, Cadre2006, 
mammen2013confidence,Laloe2012}. For a smooth density $q$, (G) holds whenever
the specified level $\lambda$ does not coincide with
the density value for a critical point.

Assumption (K1) is to guarantee that the variance of
the KDE is bounded and to ensure that $p_h\in\BC^3$.
This assumption is very common in statistical literature,
see e.g. \cite{wasserman2006all}. 
Assumption (K2) is to regularize the complexity of the kernel function
so that the supremum norm for kernel functions and their derivatives 
can be bounded in probability.
Similar assumption appears in \cite{Einmahl2005} and \cite{genovese2014nonparametric}.
The Gaussian kernel and many compactly supported kernels satisfy both assumptions.

An immediate result from assumption (G) is the smoothness of the density
level set. This smoothness property will be used to understand
the distribution of $\Haus(\hat{D}_h, D_h)$.


\begin{lem}[Smoothness Theorem]
Assume a density $p_h\in\BC^2$ satisfies condition (G) and let $D_h$ denote
the level set for $p_h$ at $\lambda$.
Then
\begin{equation}
\reach(D_h) \geq \min\left\{\frac{\delta_0}{2}, \frac{g_0}{\norm{p_h}^*_{2,\max}}\right\}.
\end{equation}
Moreover, let $q\in\mathbf{BC}^3$ be another density function and define $D(q)$ as the level set
for $q$ at level $\lambda$. When $\norm{p_h-q}^*_{2,\max}$ is sufficiently small,
\begin{itemize}
\item[1.] Condition (G) holds for $q$.
\item[2.] $\reach(D(q)) = \min\left\{\frac{\delta_0}{2}, \frac{g_0}{\norm{p_h}^*_{2,\max}}\right\}
+O(\norm{p_h-q}^*_{2,\max})$.
\item[3.] $D_h$ and $D(q)$ are normal compatible. 
\end{itemize}
\label{lem::normal}
\end{lem}

The proof is given in the supplementary materials.
Lemma \ref{lem::normal} is very similar to Theorem 1 and 2 in \cite{Walther1997}.
Essentially, this lemma shows that the level set $D$ is smooth
and that whenever two smooth densities are sufficiently close,
their level sets will both be smooth, be close to each other, and
have one-to-one and onto normal projections between them.

Given a collection of functions $\mathcal{F} = \{f_t:\R^d\mapsto \R: t\in T\}$, where $T$ is some
index set, the empirical process $\mathbb{G}_n$
is defined as 
\begin{equation}
\mathbb{G}_n(f) = 
\frac{1}{\sqrt{n}}\sum_{i=1}^n (f(X_i)-\mathbb{E}(f(X_1))),\quad f\in \cF.
\end{equation}

\begin{lem}[Empirical Approximation]
Assume (K1--K2) and (G) hold for $p_h$.
Let $D_h$ and $\hat{D}_h$ be the density level sets with level $\lambda$
for $p_h$ and $\hat{p}_h$.
Define the function
\begin{equation}
f_x(y) = \frac{1}{\sqrt{h^d}\norm{ \nabla p_h(x)}}K\left(\frac{x-y}{h}\right)
\end{equation}
with $x\in D_h$.
If $\frac{\log n}{nh^{d+2}} \rightarrow 0, \, h\rightarrow 0$, then
\begin{equation}
\sup_{x\in D_h}\left|\frac{\mathbb{G}_n(f_x)-\sqrt{nh^d} \cdot d(x, \hat{D}_h)}{\sqrt{nh^d}\cdot d(x, \hat{D}_h)} \right|
= O(\norm{\hat{p}_h-p_h}^*_{1,\max}) = O_P\left(\sqrt{\frac{\log n}{nh^{d+2}}}\right).
\end{equation}
\label{thm::empirical}
\end{lem}

A key element for the proof of Lemma~\ref{thm::empirical}
is the smoothness of $D_h$ and $\hat{D}_h$, which relies on Lemma~\ref{lem::normal}.
This smoothness  allows us to 
approximate the local difference by an empirical process.

\begin{remark}
Lemma~\ref{thm::empirical} implies that 
for each $x\in D_h$,
$d(x,\hat{D}_h)$
converges to a mean $0$ Gaussian process.
We can use $\mathbb{E}\left(d(x,\hat{D}_h)^2\right)$ as a measure
of \emph{local uncertainty} \cite{chen2014asymptotic} -- analogous to the mean squared error --
and thus can apply the bootstrap to estimate this quantity.
\end{remark}

Lemma~\ref{thm::empirical} shows that 
the projected distance to
the level sets can be approximated by a stochastic process 
(the empirical process) defined on a smooth manifold.
Specifically,
Lemma~\ref{thm::empirical} shows that 
the projection distance can be approximated by 
an empirical process on certain functions $f_x$,
where $x\in D_h$. The level sets $D_h$ now acts as an index set.
Thus, we define the function space
\begin{equation}
\cF = \left\{f_x(y)\equiv\frac{1}{\sqrt{h^d}\norm{ \nabla p_h(x)}}K\left(\frac{x-y}{h}\right): x\in D_h\right\}
\label{eq::cF}
\end{equation}
and define a Gaussian process $\mathbb{B}$ on $\cF$
such that for all $f_1, f_2\in\cF$,
\begin{equation}
\mathbb{B}(f_1) \overset{d}{=} N(0, \mathbb{E}(f_1(X_1)^2)),\quad \Cov(\mathbb{B}(f_1), \mathbb{B}(f_2)) 
= \mathbb{E}(f_1(X_1)f_2(X_1)).
\label{eq::G}
\end{equation}


\begin{thm}[Asymptotic Theory]
Assume (K1--K2) and (G) holds for $p_h$.
Let $D_h$ and $\hat{D}_h $ be the density level sets with level $\lambda$
for $p_h$ and $\hat{p}_h$. 
Then when $\frac{\log n}{nh^{d+2}} \rightarrow 0,\, h\rightarrow 0$, 
the Hausdorff distance satisfies
\begin{equation}
\sup_t\left|\mathbb{P}\left(\sqrt{nh^d}\ \Haus(\hat{D}_h, D_h)<t\right) - 
\mathbb{P}\left(\sup_{f\in\cF}|\mathbb{B}(f)|<t\right)\right| 
= O\left(\left(\frac{\log^7 n}{nh^d}\right)^{1/8}\right),
\end{equation}
where $\cF$ is defined in equation \eqref{eq::cF}
and $\mathbb{B}$ is a Gaussian process defined on $\cF$ satisfying equation \eqref{eq::G}.
\label{thm::Gaussian}
\end{thm}
The proof of Theorem~\ref{thm::Gaussian} depends on two geometric observations.
First, the empirical approximation in Lemma~\ref{thm::empirical},
shows that the local difference is approximately the same as an empirical process,
and hence the maximum local difference is approximated by the maximum of the empirical process.
Second, the normal compatibility between $\hat{D}_h$ and $D_h$ guaranteed by 
Lemma~\ref{lem::normal},
which implies that maximal of local difference is the same as the Hausdorff distance. 

Theorem \ref{thm::Gaussian} shows that the Hausdorff distance $\Haus(\hat{D}_h, D_h)$
can be approximated by a maximum over a certain Gaussian process.
Note that we cannot directly use this theorem to construct a confidence set
for $D_h$ since the Gaussian process is defined on $D_h$, which is unknown.
Later we will use the bootstrap to approximate this limiting distribution and
construct a confidence set.

The random variable $\sup_{f\in\cF}|\mathbb{B}(f)|$
follows an extreme-value type distribution.
However, writing down the explicit form for this distribution
is not very helpful in statistical inference because it involves 
unknown quantities and because the convergence 
to the distribution is notoriously slow.
Instead, we will use the bootstrap to approximate the 
distribution of $\Haus(\hat{D}_h, D_h)$,
avoiding the unknown quantities and yielding much faster convergence.

\section{Statistical Inference}	\label{sec::inf}


We now show that we can construct valid
confidence sets for $D_h$ with the bootstrap.
A set $S_{n, 1-\alpha}$ is called an asymptotically valid
confidence set for $D_h$ if
\begin{equation}
\mathbb{P}(D_h \subset  S_{n, 1-\alpha}) = 1-\alpha + O(r_n),
\end{equation}
where $r_n\rightarrow 0$ as $n\rightarrow \infty$.
We propose two methods for constructing a confidence set,
and we will show that they are both asymptotically valid.

\subsection{Method 1: Hausdorff Loss}

The first approach is to use the Hausdorff distance between the level sets.
Let $W_n = \Haus(\hat{D}_h, D_h)$ and define
\begin{equation}
w_{1-\alpha} = F^{-1}_{W_n}(1-\alpha),
\end{equation}
where $F_A$ denotes the cdf for a random variable $A$.
Then, it is easy to see that
\begin{equation}
\mathbb{P}(D_h\subset \hat{D}_h\oplus w_{1-\alpha})\geq 1-\alpha.
\end{equation}
We use the bootstrap to estimate $w_{1-\alpha}$.

Let $X^*_1,\cdots, X^*_n$ be a bootstrap sample from the empirical measure.
Let $\hat{p}_h^*$ denote the KDE using the bootstrap sample,
and $\hat{D}^*_n$ the corresponding level set.
We define $W^*_n = \Haus(\hat{D}^*_n, \hat{D}_h)$
and 
\begin{equation}
\hat{w}_{1-\alpha} = F^{-1}_{W^*_n}(1-\alpha),
\end{equation}
Then the bootstrap confidence set is $\hat{D}_h\oplus \hat{w}_{1-\alpha}$.

\begin{thm}
Assume (K1--K2) and (G) holds for $p_h$ and 
$\frac{\log n}{nh^{d+2}} \rightarrow 0, \,h\rightarrow 0$.
Let $D_h$ and $\hat{D}_h $ and $\hat{D}^*_n$ be the 
density level set with level $\lambda$
for $p_h$ and $\hat{p}_h$ and $\hat{p}_h^*$. Then
there exist $\mathcal{X}_n$ such that
\begin{align}
\sup_t\Big|\mathbb{P}\left(\sqrt{nh^d}\Haus\left(\hat{D}^*_n, \hat{D}_h\right)<t \,\big|\,X_1,\cdots,X_n\right)&- 
         \mathbb{P}\left(\sqrt{nh^d}\Haus(\hat{D}_h, D_h)<t\right)\Big| \nonumber\\
   &= O\left(\left(\norm{\hat{p}_h-p_h}_{\max}\right)^{1/8}\right)
\end{align}
for all $(X_1,\cdots,X_n)\in \mathcal{X}_n$
and $\mathbb{P}(\mathcal{X}_n)\geq 1-3e^{-nh^{d+2}\tilde{A}_0}$ for some constants $\tilde{A}_0$.
Thus,
\begin{equation}
\mathbb{P}\left(D_h\subset \hat{D}_h\oplus \hat{w}_{1-\alpha}\right) =
    1-\alpha+O\left(\left(\frac{\log^7 n}{nh^{d}}\right)^{1/8}\right).
  \end{equation}
\label{thm::bootstrap}
\end{thm}

An intuitive explanation for Theorem~\ref{thm::bootstrap}
is that as $n$ goes to infinity, the bootstrap process
converges to the same Gaussian process as the empirical process -- 
thus, they share the same Berry-Esseen bound.

\subsection{Method 2: Supremum Loss}\label{sec::CI::M2}

The second approach is to use the supremum norm of the KDE
and impose an upper and lower bound around the density level.

Let $M_n = \sup_{x\in\K}|\hat{p}_h(x)-p_h(x)|$ and $m_{1-\alpha} = F^{-1}_{M_n}(1-\alpha)$.
Define
\begin{equation}
C_{n,1-\alpha} = \Set{x\in\K\st |\hat{p}_h(x)- \lambda|\leq m_{1-\alpha}}.
\label{eq::CI0}
\end{equation}
It is easy to verify that 
\begin{equation}
\mathbb{P}(D_h\subset C_{n,1-\alpha})\geq 1-\alpha.
\label{eq::CI_C}
\end{equation}

Again we use the bootstrap to estimate the quantile.
Recall that $\hat{p}_h^*$ is the KDE based on the bootstrap sample.
We define $M^*_n = \sup_{x\in\K}|\hat{p}_h^*(x)- \hat{p}_h(x)|$
and set
\begin{equation}
\hat{m}_{1-\alpha} = F^{-1}_{M^*_n}(1-\alpha).
\end{equation}
Then the confidence set is
\begin{equation}
\hat{C}_{n,1-\alpha} = \{x\in\K: |\hat{p}_h(x)- \lambda|\leq \hat{m}_{1-\alpha}\}.
\end{equation}

\begin{thm}
Assume (K1--K2) and (G) holds for $p_h$ and $\frac{\log n}{nh^{d+2}} \rightarrow 0, \,h \rightarrow0$.
Let $D_h$ and $\hat{D}_h $ and $\hat{D}^*_n$ be the density level sets with level $\lambda$
for $p_h$ and $\hat{p}_h$ and $\hat{p}_h^*$. Then
$$
\mathbb{P}\left(D_h\subset \hat{C}_{n,1-\alpha}\right)
= 1-\alpha+O\left(\left(\frac{\log^7 n}{nh^d}\right)^{1/8}\right).
$$
\label{thm::CI2}
\end{thm}
The proof of this Theorem is similar to the proof of 
Theorem~\ref{thm::bootstrap}, so we omit the details.
The basic idea is as follows.
By equation \eqref{eq::CI_C},
the quantile of $\hat{M}_n$ can be used to construct confidence sets.
We then show that $\sqrt{nh^d}{M}_n$ is approximated by the maximum of a Gaussian process 
(similar to Theorem~\ref{thm::Gaussian} and made explicit in \citealt{chernozhukov2014anti}).
Finally, we show that the bootstrap $\sqrt{nh^d}M^*_n$ converges to $\sqrt{nh^d}{M}_n$
as in Theorem~\ref{thm::bootstrap}.

This method embodied in Theorem \ref{thm::CI2} is very similar to the methods in
\cite{jankowski2012confidence} and \cite{mammen2013confidence}.
\cite{jankowski2012confidence}
proposes to construct a confidence set of the form
\begin{equation}
C_n = \{x: \lambda-\ell_n\leq\hat{p}_h(x)\leq \lambda+\tau_n\}
\end{equation}
with some $\ell_n, \tau_n\rightarrow 0$.
They require
that $\sqrt{nh^d}(\hat{p}_h-p_h)$ converges weakly
to a known random field.
This convergence is true when $h$ is fixed but is \emph{not} attainable if we allow $h=h_n\rightarrow 0$,
in contrast to our approach, which supports both cases.
They also assume the limiting random field is either known or can be easily estimated;
we do not require that assumption.

\cite{mammen2013confidence}
construct a confidence set using a similar approach to method 2,
but they focus on the original level set $D = \{x: p(x)=\lambda\}$ 
rather than the smoothed version $D_h$.
Instead of taking the supremum of the density deviation over the whole
support $\K$, they propose to focus on the regions $x\in D\triangle \hat{D}_h$. i.e.
\begin{equation}
R_n = \sup_{x\in D\triangle \hat{D}_h} |\hat{p}_h(x)-p|,
\end{equation}
where $A\triangle B= \{x: x\in A, x\notin B\}\cup \{x: x\in B, x\notin A\}$ is the symmetric
difference between sets.
Then they use the upper quantile of $R_n$ to construct a confidence set
of a similar form to \eqref{eq::CI0}
and apply the bootstrap to estimate the quantile.
Their bootstrap consistency relies on
Neumann's method (Proposition 3.1 in \cite{neumann1998strong})
and they assume that $h$ converges fast enough so that one can ignore the
bias for estimating the original density $p$.
Actually, under their assumptions, our proposed bootstrap confidence sets
(from both methods 1 and 2) are also consistent
for the original level set $D$ since the bias converges faster than
the stochastic variation.
The method in \cite{mammen2013confidence} 
should have higher power than our method 2
since they consider taking the supremum over a smaller region.

\begin{remark}	\label{rm::scaled}
We may use a variance stabilizing transform to obtain an adaptive
confidence set using similar idea to \cite{chernozhukov2012inference}.
The variance of $\hat{p}_h(x)$ is proportional to $p(x)$.
Thus, we may use
\begin{align}
V^*_n &= \sup_{x\in\K}\frac{|\hat{p}_h^*(x)- \hat{p}_h(x)|}{\sqrt{\hat{p}_h(x)}}\\
\hat{v}_{1-\alpha} &= F^{-1}_{V^*_n}(1-\alpha)
\end{align}
and set $\hat{v}_{1-\alpha}\times\sqrt{\hat{p}_h(x)}$
as an adaptive threshold for constructing the confidence set.
Namely, the adaptive confidence set is given by
\begin{equation}
\hat{C}^*_{n,1-\alpha} = 
    \Set{x\in\K\st |\hat{p}_h(x)- \lambda|\leq \hat{v}_{1-\alpha}\times\sqrt{\hat{p}_h(x)}}.
\end{equation}
Using the same approach as in the proof to Theorem~\ref{thm::CI2},
we can show that $\hat{C}^*_{n,1-\alpha}$ has asymptotically $1-\alpha$ coverage.
\end{remark}

\begin{remark}
The rate $O\left(\left(\frac{\log^7 n}{nh^d}\right)^{1/8}\right)$ may not be optimal.
In \cite{chernozukov2014central}, they apply a induction technique that 
gives a rate of order $n^{-1/6}$ for the Gaussian approximation.
Despite not being mentioned explicitly in that paper,
we believe that similar technique applies to the empirical process.
If this is true, the rate in Theorem~\ref{thm::Gaussian},
\ref{thm::bootstrap} and \ref{thm::CI2}
can be further refined to 
$O\left(\left(\frac{\log^7 n}{nh^d}\right)^{1/6}\right)$.

\end{remark}

\subsection{Comparing Methods 1 and 2} 

Both methods 1 and 2 generate confidence sets with asymptotically valid coverage.
Figure~\ref{fig::vis_two} compares the $90\%$ confidence sets from both methods
on the old faithful dataset.
Apparently, method 1 (left panel; blue regions) is superior to
method 2 (right panel; gold regions)
in the sense that the size of confidence set is much smaller.
The main reason is that both methods use the maximum over certain
empirical processes but the two processes are defined on different function spaces.
Method 1 only takes the supremum over a small function space $\cF$,
in which the index set contains only points on the level sets $D_h$.
However, method 2 takes the maximum over a large function space
whose index set is the whole space $\K$.
Thus, we expect the second method to have a wider confidence set.

Despite the fact that method 2 yields a much larger confidence sets,
it has some nice properties.
First, method 2 is very simple: all we need is to compute
the bootstrap distribution of supremum loss.
Second, method 2 is connected to the methods proposed in 
\cite{mammen2013confidence} and \cite{jankowski2012confidence}.
Third, the confidence sets produced in method 2 are
related to level sets with level $\lambda\pm \hat{m}_{1-\alpha}$.
The last property makes it easy to visualize the confidence sets;
see Section \ref{sec::vis::CI}.


\begin{figure}
\centering
\includegraphics[width=2in]{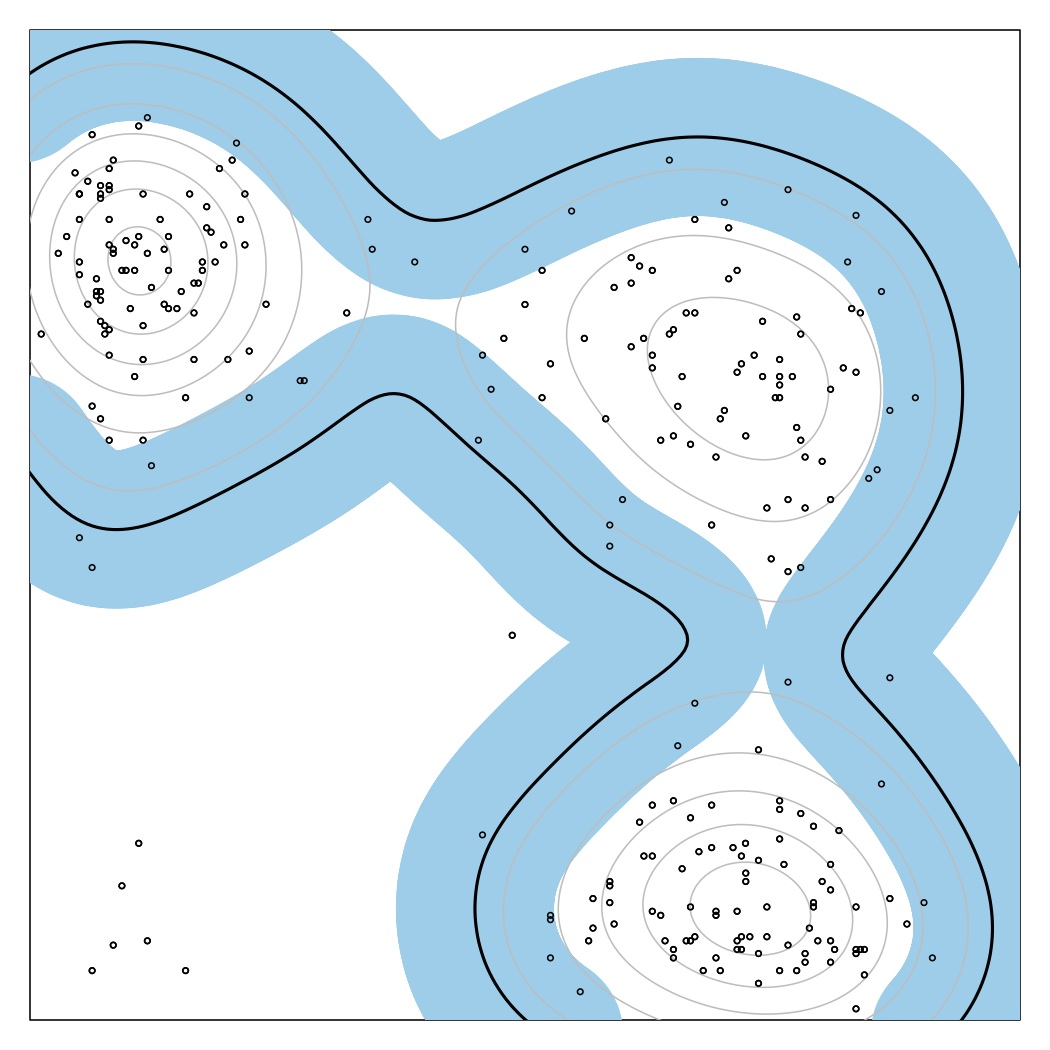}
\includegraphics[width=2in]{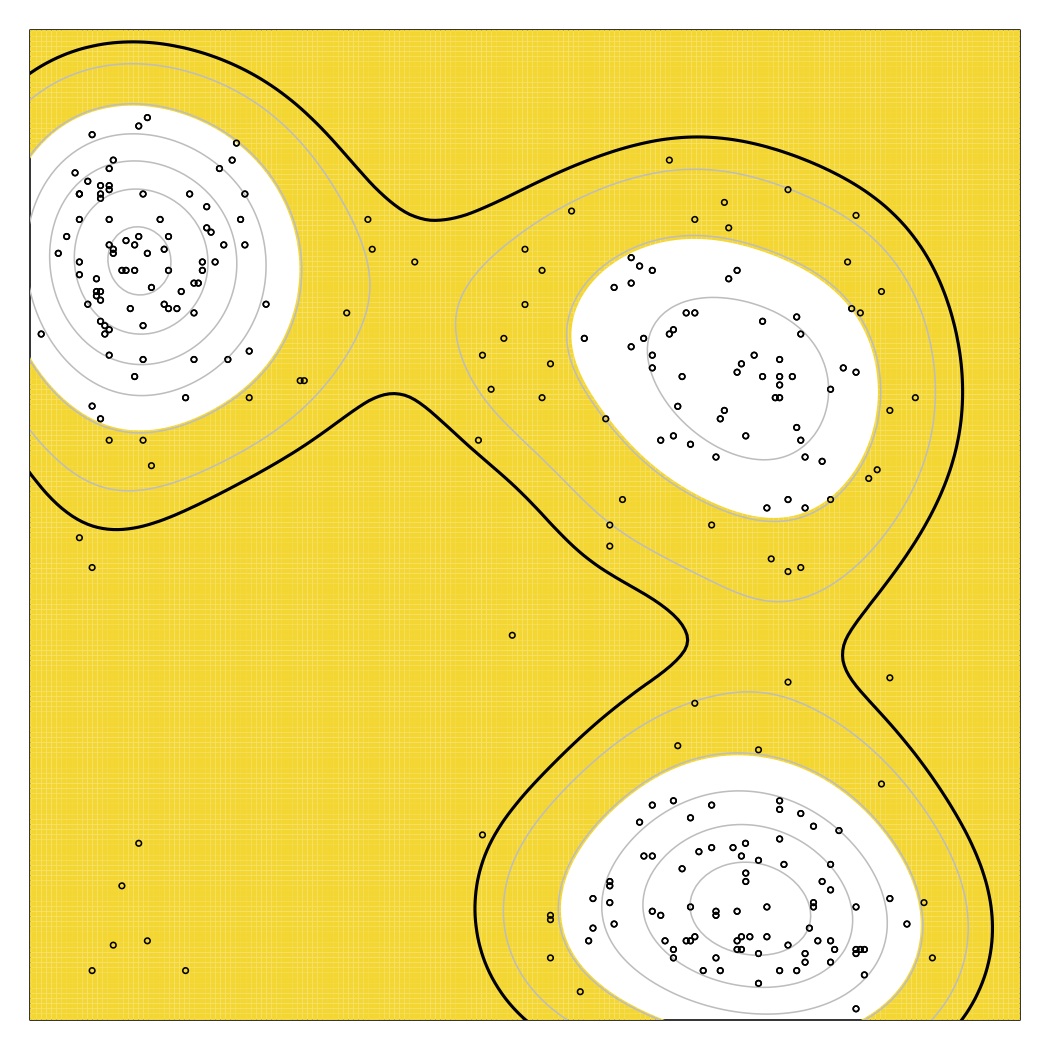}
\caption{An example of $90\%$ confidence sets 
using Hausdorff loss of level set (method 1; blue regions in left panel) and the supremum
loss of the density (method 2; gold regions in right panel).
This dataset is the old faithful dataset.
As can be seen easily, the supremum loss of density (right panel) is too huge
so that it contains a wide regions as the confidence set.
On the other hand, Hausdorff loss of level sets (left panel) gives
a much tighter confidence set.
The smoothing parameter $h=0.3$ and the two axes are from $1.5$ to $5.5$.}
\label{fig::vis_two}
\end{figure}

\begin{figure}
\centering
\includegraphics[width=2in]{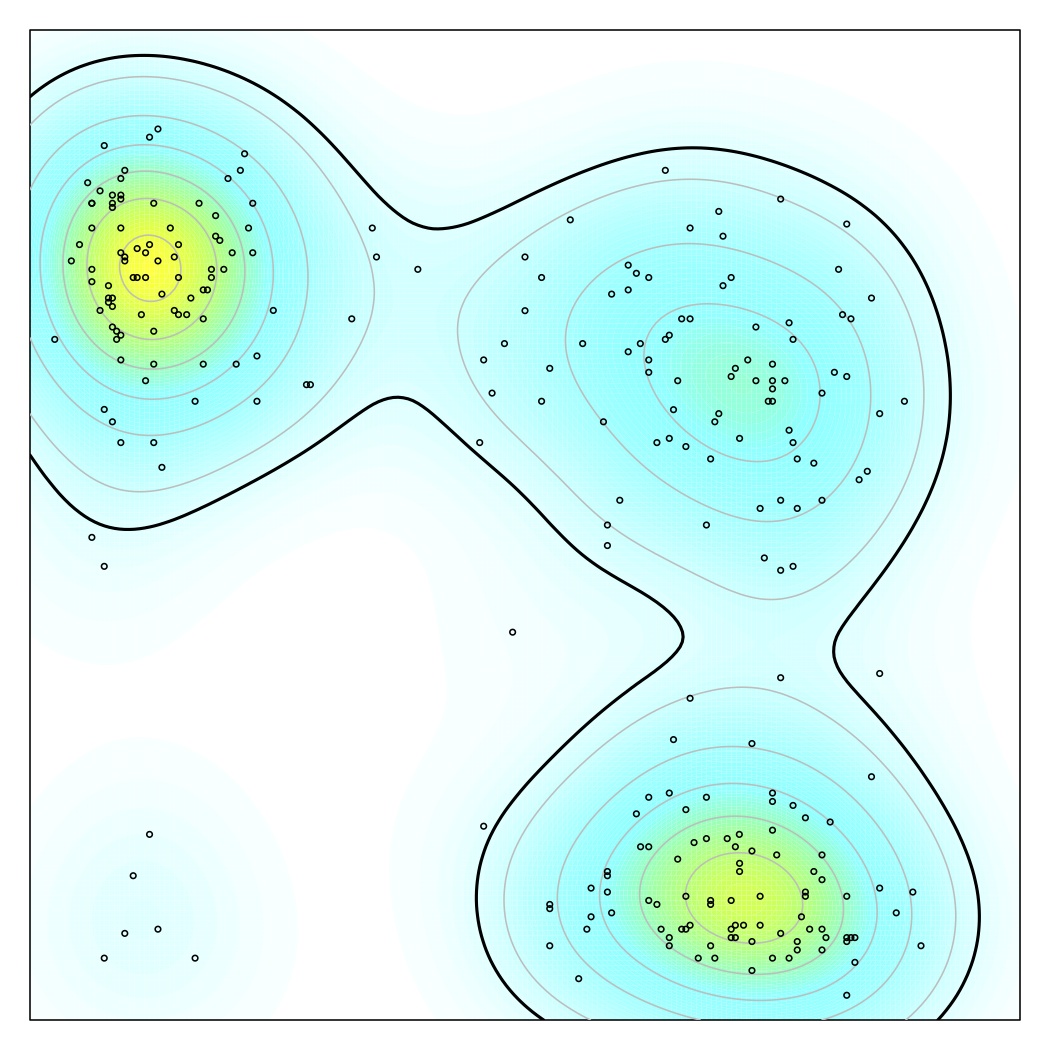}
\includegraphics[width=2in]{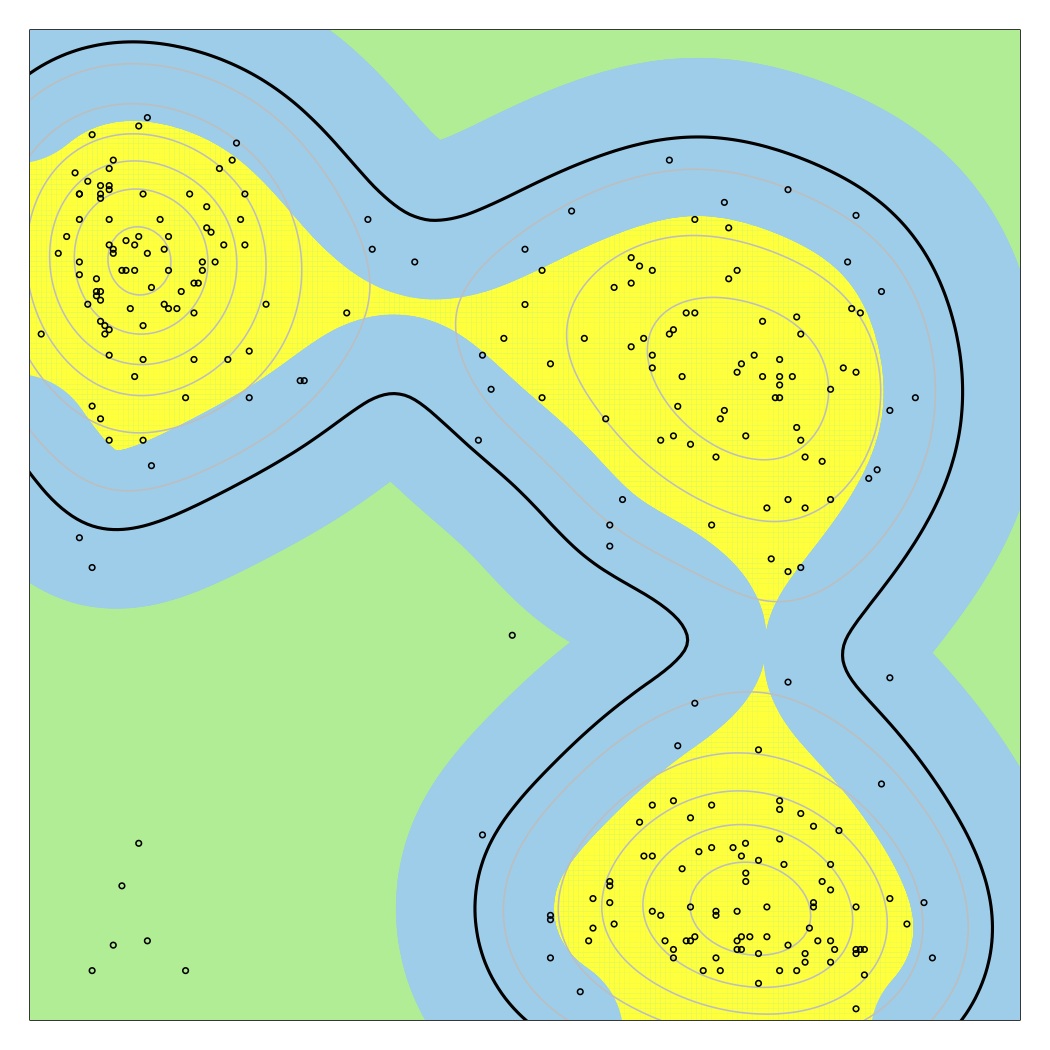}
\caption{An example for density level sets and confidence regions for the old faithful data.
Left: density level set (thick black contour denotes the specified level $\lambda$).
Right: $90\%$ confidence regions for the density level sets.
We also have $90\%$ confidence that (1) all the yellow regions have density above $\lambda$
(2) all green regions have density below $\lambda$ and
(3) the level sets $\{x: p_h(x)= \lambda\}$ are within the blue regions.
Note that the yellow and green regions are the collection of $x$ 
that we reject $H_{\sf in, 0}(x)$
and $H_{\sf out, 0}(x)$ and simultaneously control the significance level at $\alpha=10\%$.}
\label{fig::CR}
\end{figure}

\begin{figure}
\centering
\includegraphics[width=2in]{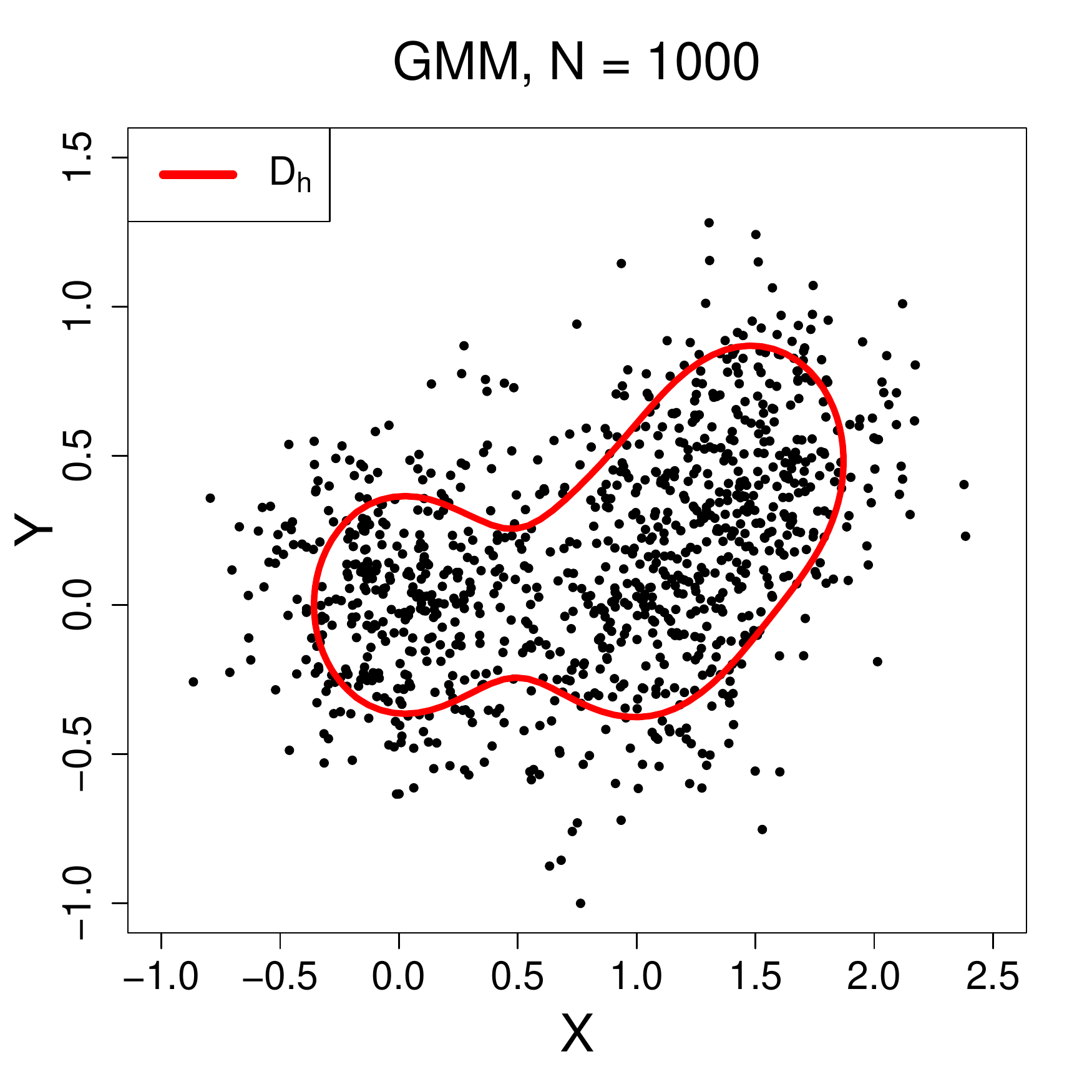}
\includegraphics[width=2in]{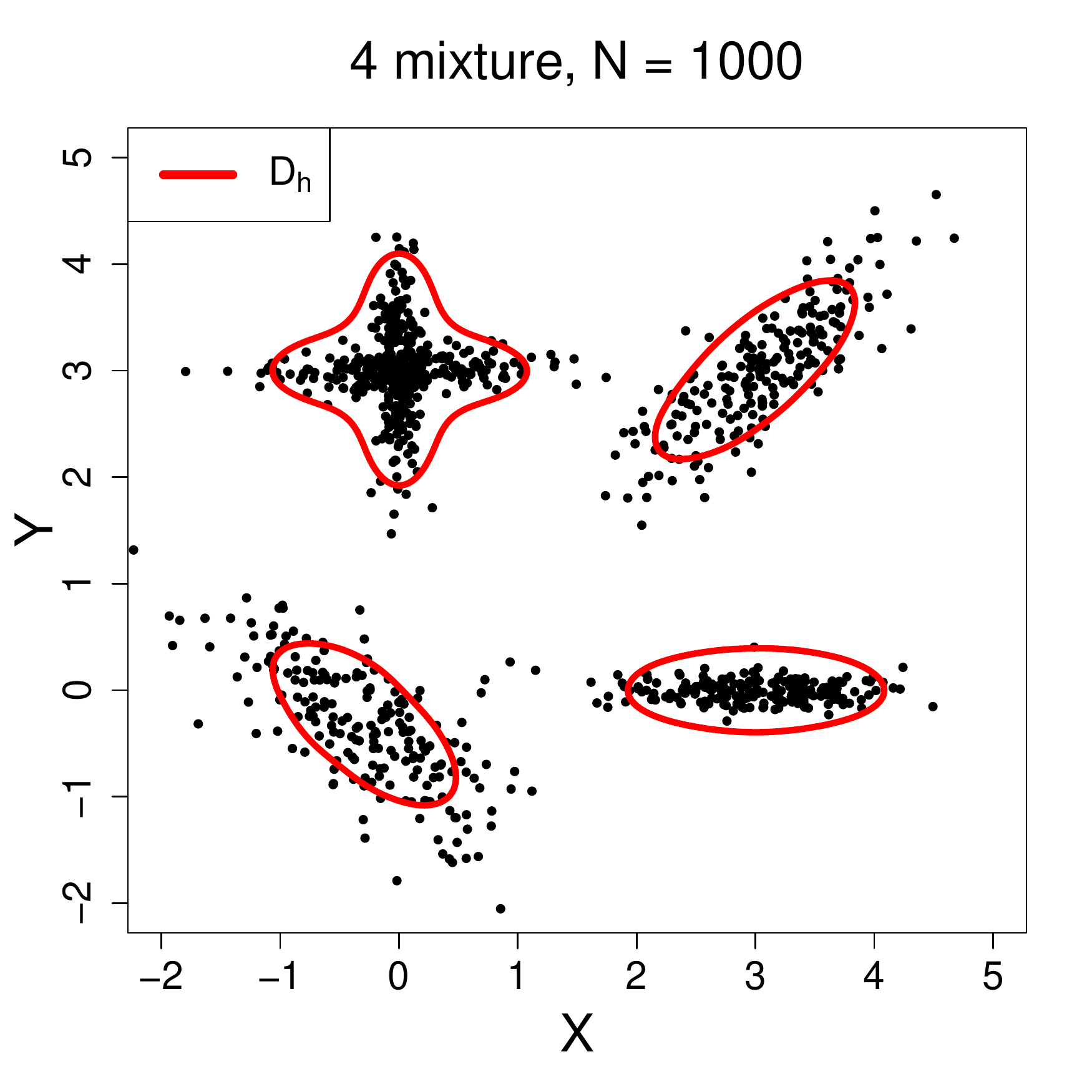}
\caption{Simulated data for comparing coverage of confidence sets.
Left: Three Gaussian mixture dataset.
Right: Four-mixture dataset with sample from \cite{cadre2009clustering}. 
The red curves indicate the density level set we are interested in.}
\label{fig::sim}
\end{figure}

Here, we consider two simulated datasets
to compare the coverage for confidence sets
constructed using Hausdorff loss (method 1),
$L_\infty$ loss (supremum loss; method 2),
and scaled $L_\infty$ loss (remark \ref{rm::scaled}).

The first dataset is a three-Gaussian mixture.
The data is generated from the following distribution:
\begin{equation}
p(x) = \frac{1}{3} \phi_2(x;\mu_1, 0.3^2 \mathbf{I}_2)
+\frac{1}{3} \phi_2(x;\mu_2, 0.3^2 \mathbf{I}_2)
+\frac{1}{3} \phi_2(x;\mu_3, 0.3^2 \mathbf{I}_2),
\end{equation}
where $\phi_2(x;\mu, \Sigma)$ is the density to bivariate Gaussian with mean
vector $\mu$ and Covariance matrix $\Sigma$.
$\mathbf{I}_2$ is the $2\times 2$ identity matrix
and $\mu_1=(0,0)^T$, $\mu_2=(1,0)^T$, and $\mu_3= (1.5, 0.5)^T$.
We use density level $\lambda=0.3$ and smoothing parameter $h=0.2$.
The corresponding level set $D_h$ is the red curve in the left panel of Figure~\ref{fig::sim}.
We consider three different sample sizes, $N=500, 1000, 2500$ and compare
the coverage of confidence sets using the three methods. 
The corresponding coverages are given in Table \ref{tab::cover_3GMM}.
As can be seen from Table \ref{tab::cover_3GMM},
all the three methods have the desire nominal coverage.

The second dataset is a four-mixture dataset from \cite{cadre2009clustering}.
The data is generated from the following distribution:
\begin{equation}
p(x) = \sum_{\ell=1}^6 \pi_\ell \phi_2(x;\mu_\ell, \Sigma_\ell)
\end{equation}
with $\pi_1=\pi_2=\pi_3=\pi_4 =1/5$ and $\pi_5=\pi_6=1/10$,
and
$\mu_1= (-0.3, -0.3)^T, \mu_2= (3.0, 3.0)^T, \mu_3 = \mu_4= (0,3)^T, \mu_5=\mu_6 = (3,0)$,
and
\begin{align}
\Sigma_1 = 
\left( \begin{array}{cc}
0.39 & -0.28 \\
-0.28 & 0.39  \end{array} \right),&\quad 
\Sigma_2 = 
\left( \begin{array}{cc}
0.36 & 0.30 \\
0.30 & 0.36  \end{array} \right),\\
\Sigma_3=\Sigma_5=
\left( \begin{array}{cc}
0.33 & 0 \\
0 & 0.01  \end{array} \right),&\quad 
\Sigma_4=\Sigma_6=
\left( \begin{array}{cc}
0.01 & 0 \\
0 & 0.33  \end{array} \right).
\end{align}
We use density level $\lambda=0.05$ and smoothing parameter $h=0.2$.
The corresponding level set $D_h$ is the red curves in the right panel of Figure~\ref{fig::sim}.
Again, we consider three different sample sizes $N=500, 1000, 2500$ and compare
the coverage of confidence sets using the three methods. 
The corresponding coverages are given in Table \ref{tab::cover_4mix}.
It is clear that
all the three methods have the desire nominal coverage.
Moreover, it can be seen from both Table \ref{tab::cover_3GMM} and \ref{tab::cover_4mix}
that the supremum loss method over covers.

\begin{table}[]
\centering
\caption{Coverage for the three-Gaussian mixture data.}
\label{tab::cover_3GMM}
\begin{tabular}{l|ll|ll|ll}
\hline
\multicolumn{1}{c|}{\multirow{2}{*}{Sample Size}} & \multicolumn{2}{c|}{Hausdorff Loss} & \multicolumn{2}{c|}{$L_\infty$ Loss} & \multicolumn{2}{c}{Scaled $L_\infty$ Loss} \\ \cline{2-7} 
\multicolumn{1}{c|}{}                             & $\alpha=0.90$    & $\alpha=0.95$    & $\alpha=0.90$     & $\alpha=0.95$    & $\alpha=0.90$        & $\alpha=0.95$        \\ \hline
$n=500$                                           & 0.946            & 0.983            & 0.992             & 0.998            & 0.957                & 0.984                \\ \hline
$n=1000$                                          & 0.944            & 0.969            & 0.991             & 0.997            & 0.946                & 0.974                \\ \hline
$n=2500$                                          & 0.915            & 0.969            & 0.992             & 0.998            & 0.947                & 0.978                \\ \hline
\end{tabular}
\end{table}

\begin{table}[]
\centering
\caption{Coverage for the four-mixture data.}
\label{tab::cover_4mix}
\begin{tabular}{l|ll|ll|ll}
\hline
\multicolumn{1}{c|}{\multirow{2}{*}{Sample Size}} & \multicolumn{2}{c|}{Hausdorff Loss} & \multicolumn{2}{c|}{$L_\infty$ Loss} & \multicolumn{2}{c}{Scaled $L_\infty$ Loss} \\ \cline{2-7} 
\multicolumn{1}{c|}{}                             & $\alpha=0.90$    & $\alpha=0.95$    & $\alpha=0.90$     & $\alpha=0.95$    & $\alpha=0.90$        & $\alpha=0.95$        \\ \hline
$n=500$                                           & 0.991            & 0.998            & 1.000             & 1.000            & 0.950                & 0.973                \\ \hline
$n=1000$                                          & 0.995            & 0.998            & 1.000             & 1.000            & 0.917                & 0.951                \\ \hline
$n=2500$                                          & 0.936            & 0.976            & 1.000             & 1.000            & 0.976                & 0.991                \\ \hline
\end{tabular}
\end{table}

\subsection{Pointwise Hypothesis Tests for Level Sets}


The confidence sets developed in the previous section are
related to two types of local hypothesis tests. 
For fixed density level $\lambda$ and an arbitrary point $x$,
consider the tests of whether $p(x)$ is greater or less than $\lambda$:
\begin{align}
H_{\sf in, 0}(x): p_h(x)&\leq \lambda,\quad H_{\sf in, A}(x): p_h(x)> \lambda,
\label{eq::Hin}\\
H_{\sf out, 0}(x): p_h(x)&\geq \lambda, \quad H_{\sf out, A}(x): p_h(x)< \lambda.
\label{eq::Hout}
\end{align}
When we only want to test just a few points,
we can do local tests for each point and control the family-wise error rate
to control the type 1 error rate.
Usually, however,
we are interested conducting the local test at many or even an infinite number of points (like a region),
making it difficult to control type 1 error simultaneously.

Inverting the confidence sets of the previous subsection
gives a solution to this problem.
Let $\hat{S}_{n,1-\alpha}$ be a confidence set for $D_h$
and let $L_h$ and $V_h$ be, respectively, the upper and lower lambda level sets.
Then the decision rules
\begin{equation}
\begin{aligned}
T_{\sf in, n}(x) &= 1(\hat{p}_h(x)\geq \lambda\wedge x\notin \hat{S}_{n,1-\alpha}),\\
T_{\sf out, n}(x) &= 1(\hat{p}_h(x)\leq \lambda\wedge x\notin \hat{S}_{n,1-\alpha}).
\end{aligned}
\label{eq::rule}
\end{equation}
control type 1 error simultaneously for all $x\in\K$. 
This simultaneous control is asymptotic in the sense that
$\mathbb{P}(T_{\sf out,n}(x)=1\quad\forall x\in L_h) = \alpha + O(r_n)$
for $r_n \to 0$ and similarly for $T_{\sf in,n}(x)$ and $V_h$.

In addition, inverting the regions where we cannot reject
$H_{\sf in, 0}(x)$ and $H_{\sf out, 0}(x)$ in \eqref{eq::rule} yields confidence sets
for $L_h$ and $V_h$.
In Figure~\ref{fig::CR}, a $90\%$ confidence regions for the upper level set $L_h$
is the union of yellow and blue regions ($T_{\sf out, n}(x) = 0$). And a $90\%$ confidence 
regions for $V_h$, the lower level set, is the union of green and blue regions ($T_{\sf in, n}(x)=0$).
Thus, with $90\%$ confidence,
all yellow regions are above $\lambda$ and the true high density regions
should be contained by the yellow and blue regions.

\begin{remark}
The two local tests described in \eqref{eq::Hin} and \eqref{eq::Hout} are relevant
to the problems of level-set clustering \citep{hartigan1975clustering, Polonik1995,Rinaldo2010a, Rinaldo2010b} 
and anomaly detection
\citep{desforges1998applications, breunig2000lof, he2003discovering,chandola2009anomaly}.
Rejecting $H_{\sf in, 0}(x)$ can be viewed as evidence
that $x$ belongs to a level-set cluster.
And a point $x$ where $H_{\sf out, 0}(x)$ is rejected can be viewed as anomalous.
\end{remark}

\begin{remark}
Note that one can modify the local testing procedure to 
control the False Discovery Rate \citep{benjamini1995controlling} 
rather than familywise error.
\end{remark}

\section{Visualization for Multivariate Level Sets}	\label{sec::vis}

\begin{figure}
\centering
\includegraphics[width=1.5in]{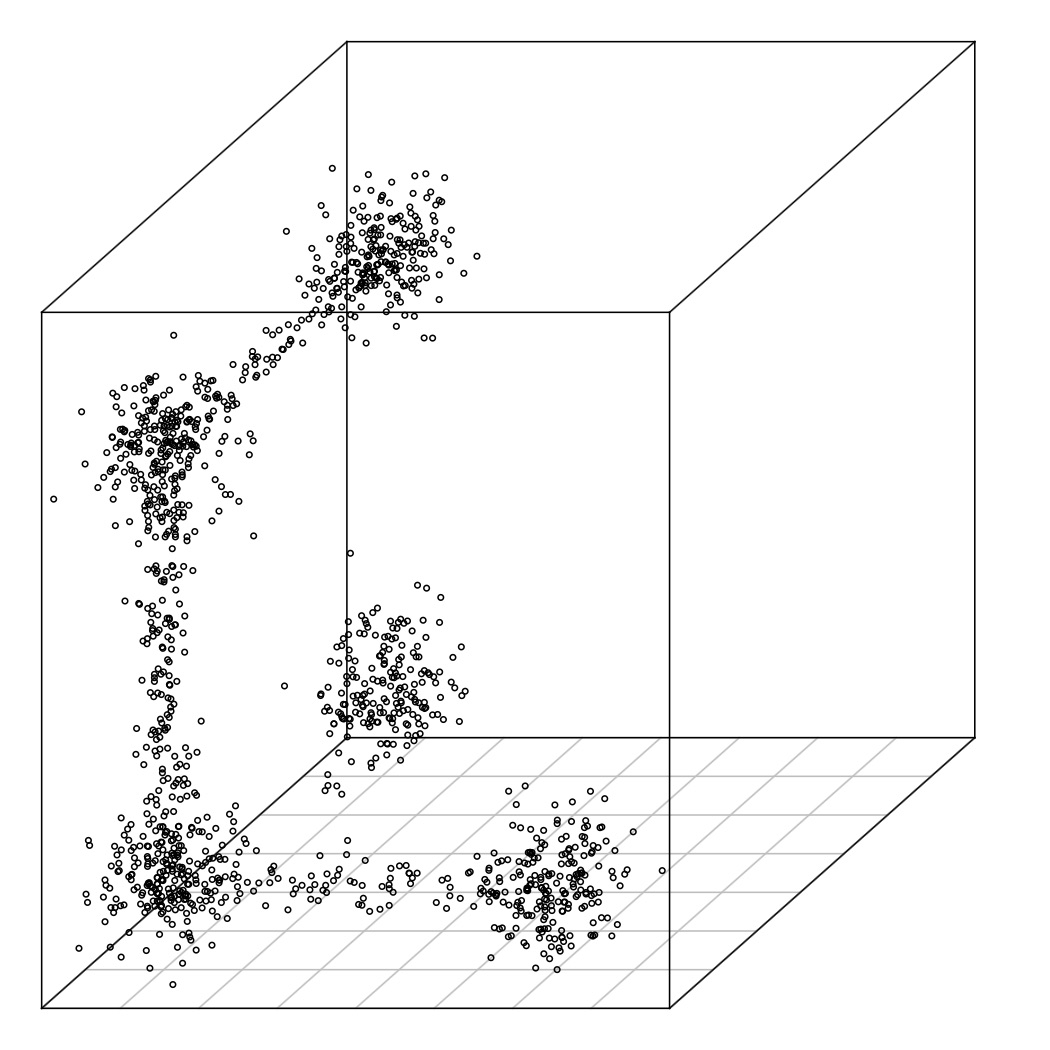}
\includegraphics[width=2in]{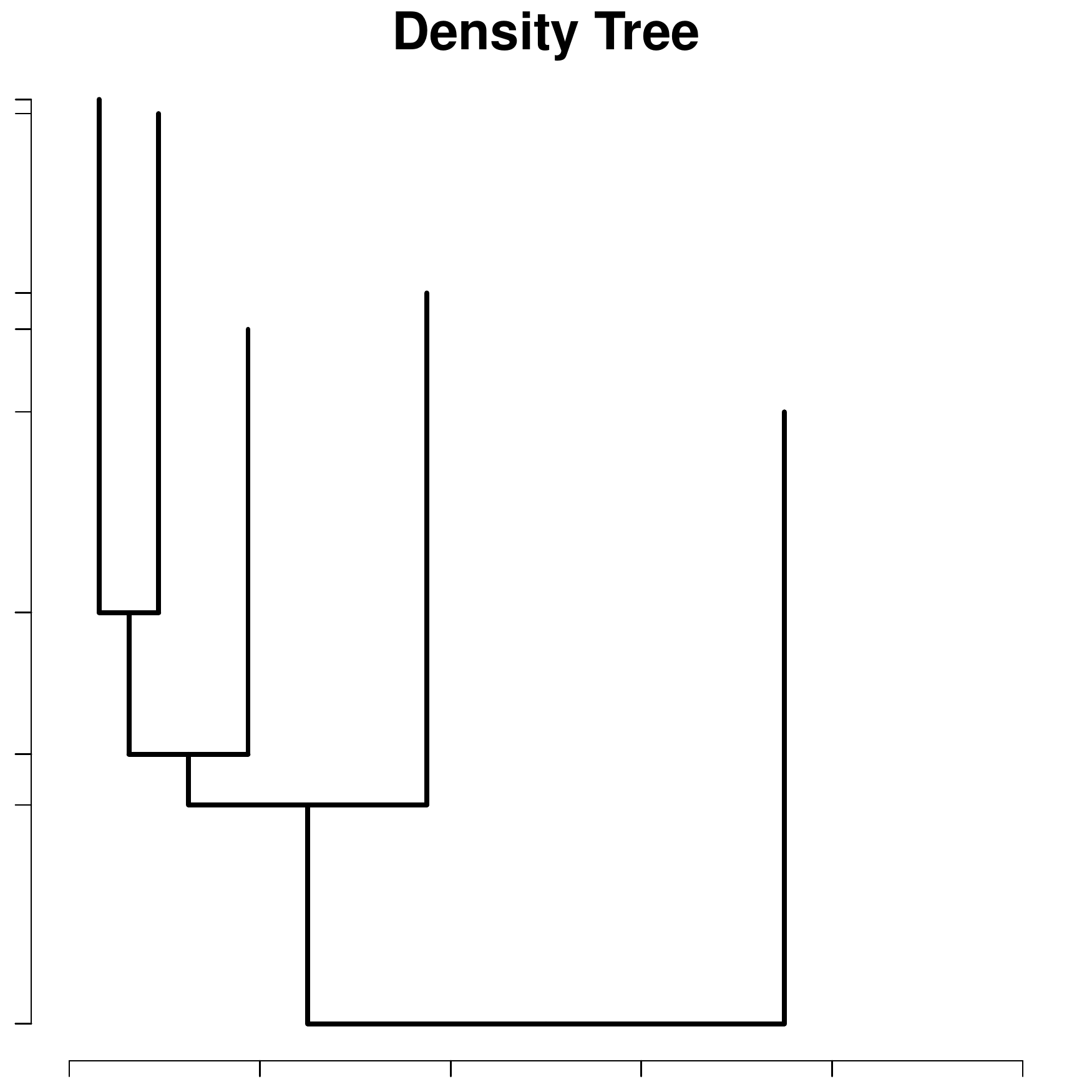}
\includegraphics[width=2in]{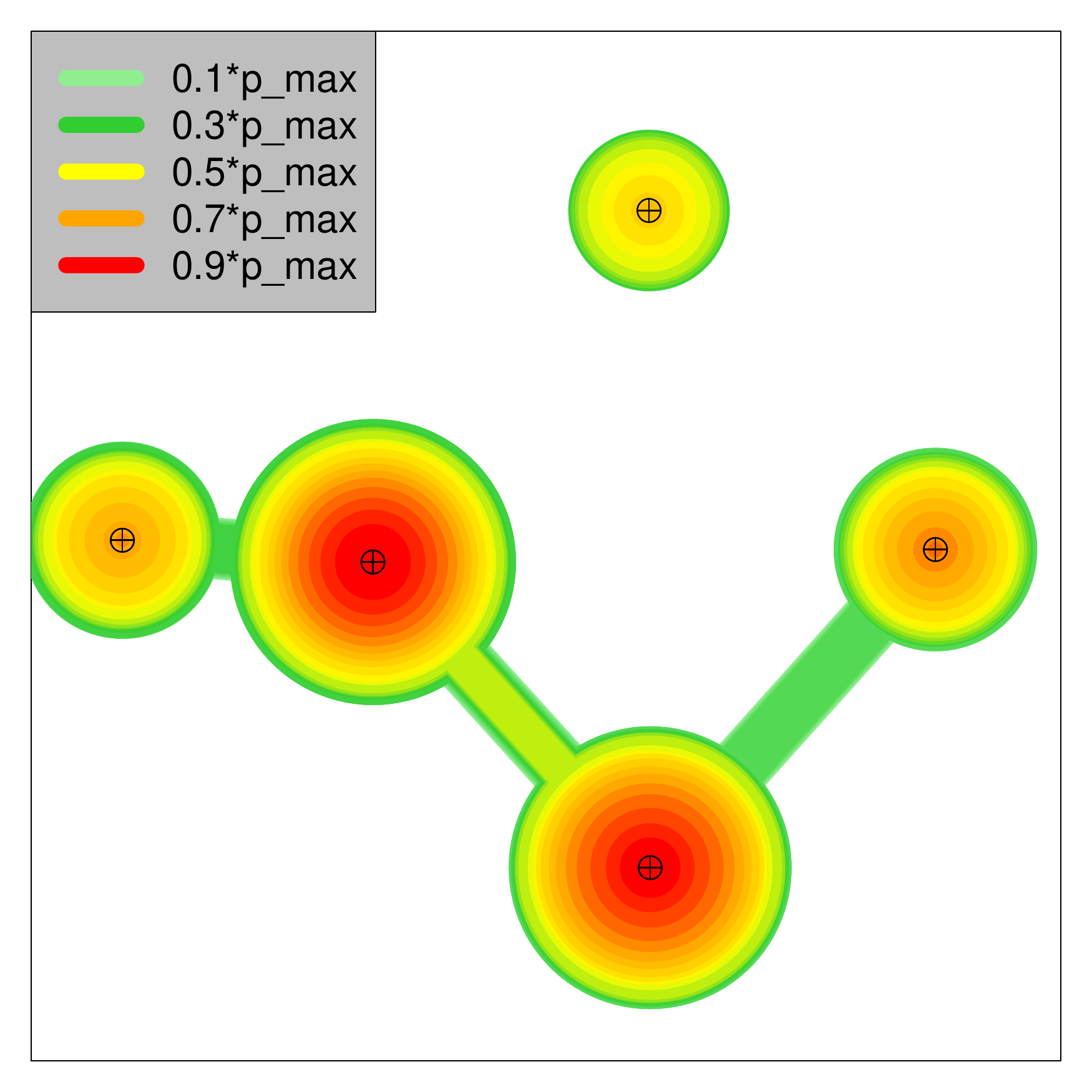}
\caption{An example for comparing density trees to our visualization method.
Left: a 3 dimensional dataset from \cite{chen2014enhanced}. 
There are 5 clusters
and four of them are connected
by a tube-like structure.
Middle: visualization using the density tree from DeBaCl \citep{kent2013debacl}.
Right: visualization using our method,
which preserves the real connections between clusters.
It is easy to see the evolution of multiple level sets
and how each cluster is connected to one another using our method.
But it is hard to get the same information
using density trees.
}
\label{fig::vis0}
\end{figure}

Level-set estimators can reveal useful information about a distribution,
but beyond three dimensions,
we cannot directly visualize the level sets, making the results
difficult to use.

In this section, we propose a novel visualization technique 
density upper level sets in multidimensions.
Any visualization entails some loss of information,
but our goals are to preserve important geometric information about the sets,
make the overall visualization easy to interpret,
and give a method that is efficient to compute.
Our method exploits the relationship between level-set clustering and mode clustering
(see Figure~\ref{fig::vis_ms}).

A current and commonly used visualization method for level sets
is the density tree \citep{stuetzle2003estimating, klemela2004visualization,
klemela2006visualization, kent2013debacl, balakrishnan2013cluster}.
This method considers several density levels, $\lambda_1<\cdots<\lambda_K$,
and computes the number of connected components for the upper density level set
at each level.
As we increase the density level, some connected components
may vanish and others may split into additional components.
In the typical case when the underlying density function is a Morse function
(i.e., Hessian at critical points is non-degenerate)
\citep{morse1925relations,morse1930foundations, milnor1963morse},
the connected component disappears
when the density level is above the maximum density value over the component
and splits only if the density level passes the density value of some saddle points
within the component \citep{klemela2009smoothing}.
The vanishing and splitting of components as the level changes produces
a tree structure.
The density tree uses this tree structure as a visualization of the level sets.
We refer to \cite{klemela2009smoothing} for more details.

Density trees display primarily topological information
about the underlying density function
\citep{stuetzle2003estimating, kent2013debacl}
but need not preserve or impart other features that may be of interest.
Extensions have been proposed that would endow the tree with additional
information about the distribution \citep{klemela2004visualization,klemela2006visualization,klemela2009smoothing},
but this is difficult to do with multiple features
and can make the visualization difficult to interpret.
See Figure~\ref{fig::vis0} for an example.

In this section, we propose a novel technique that visualizes
several density (upper) level sets that preserve some geometric information and that
are very easy to understand.
Our method is based on the relationship between level set clustering and mode clustering
(see Figure~\ref{fig::vis_ms}).
Note that in this section, we will focus on density upper level sets.

Our method complements existing tree-based methods.
It combine two clustering techniques -- level-set and mode clustering --
to produce a simple and intuitive visualization.

\begin{figure}
\centering
\includegraphics[width=1.8in]{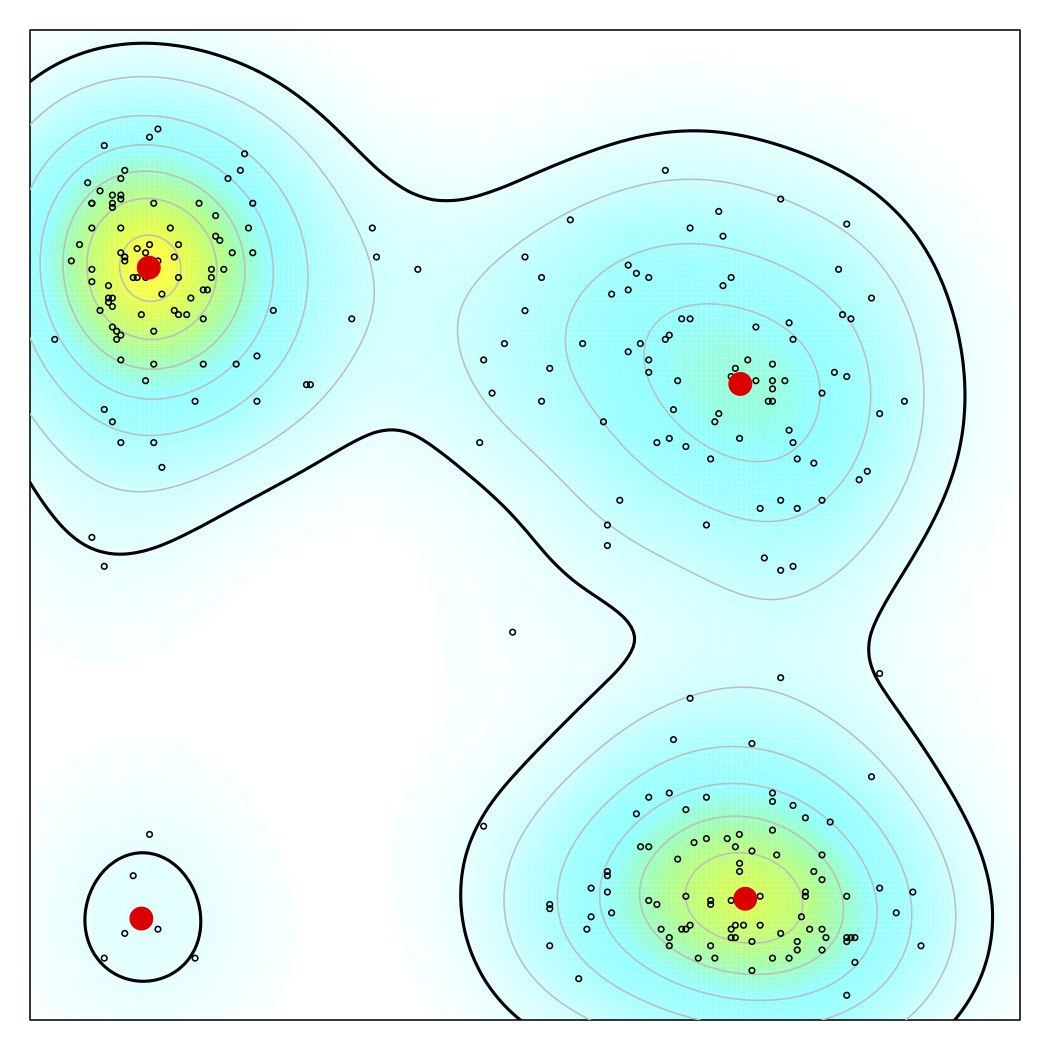}
\includegraphics[width=1.8in]{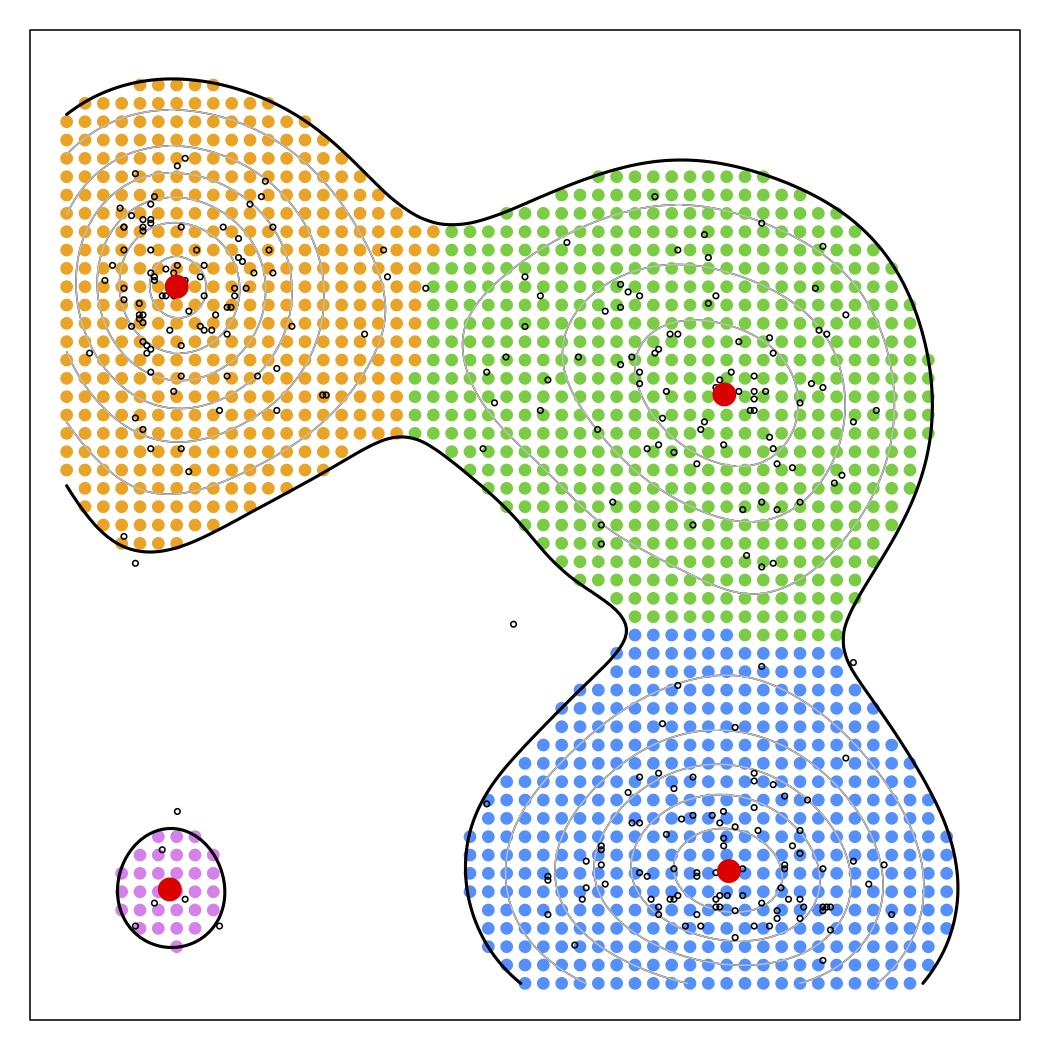}
\includegraphics[width=1.8in]{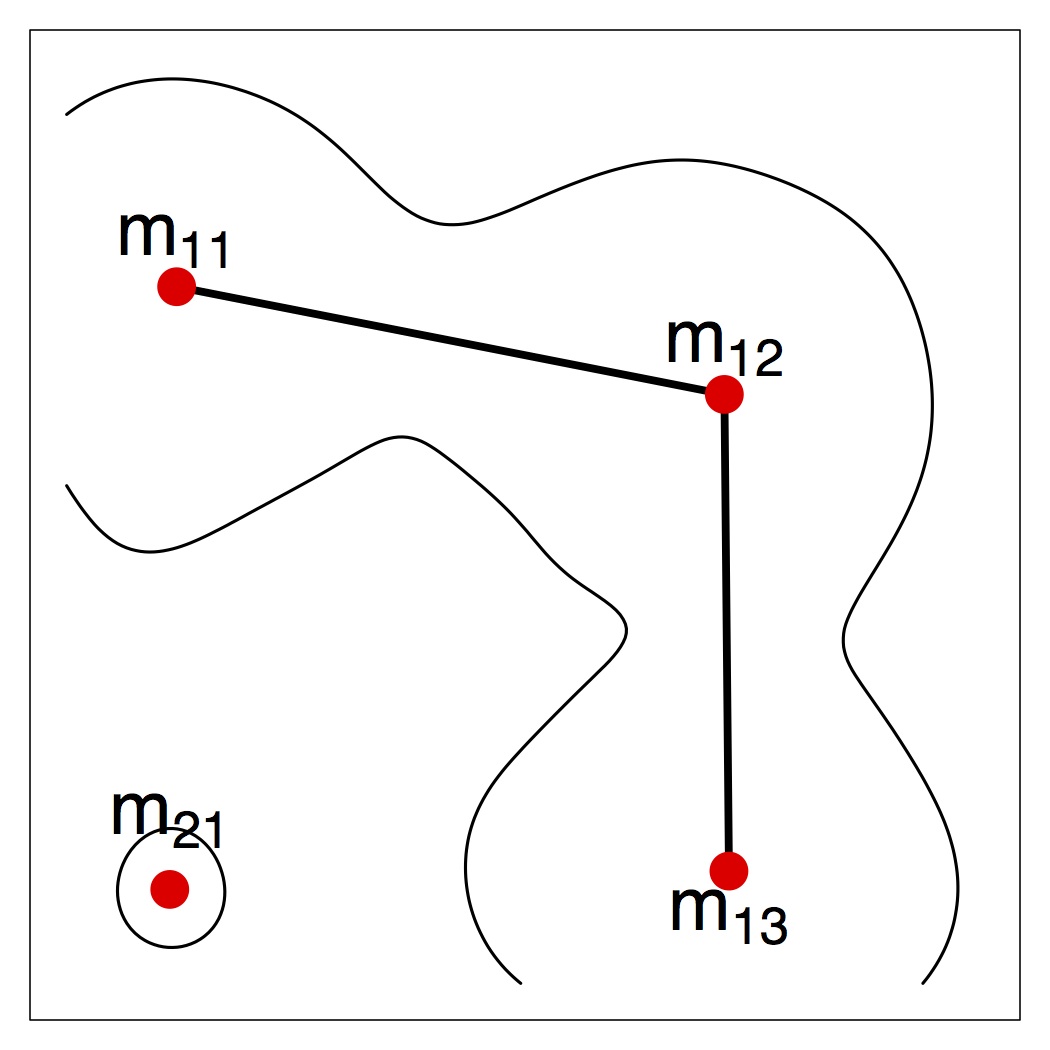}
\caption{An example for density level set and mode clustering.
Left: density level set (thick black contour denotes the specified level); four red dots
are the local modes. Middle: basins of attraction for each local modes intersecting 
the density level set. Right: graph representation. There are two connected components.
The large connected components contain three local modes $m_{1,1},m_{1,2},m_{1,3}$ 
with $m_{1,1}, m_{1,2}$ and $m_{1,2},m_{1,3}$ are connected. }
\label{fig::vis_ms}
\end{figure}

\subsection{Mode Clustering and Density Upper Level Sets}

Mode clustering (in particular, the method known as mean-shift clustering)
was introduced in \cite{fukunaga1975estimation},
where data points are clustered by following the gradient flows of the kernel density estimator
from each point to a local mode.
The clusters are the ``basins of attraction'' of each local mode \citep{fukunaga1975estimation,cheng1995mean, comaniciu2002mean}.
More generally, given a smooth Morse function $p$, mode clustering works as follows
\citep{li2007nonparametric, chacon2012clusters,chen2014enhanced}.
We form a partition of $\K$ based on the gradient field $g\equiv\nabla p$.
For each $x\in\K$, we define a gradient flow $\pi_x: [0,\infty)\mapsto \K$:
\begin{equation}
\pi_x(0) = x,\quad \pi_x'(t) = g(\pi_x(t)).
\end{equation}
That is, $\pi_x(t)$ starts at $x$ and moves along the gradient of $p$.
We define the destination for $\pi_x(t)$ as $\dest(x) = \lim_{t\rightarrow \infty}\pi_x(t)$.
Let $\cM$ be the collection of all local modes of $p$.
It can be shown that $\dest(x)\in \cM$ except for a set of $x$'s in a set
$\cB$ with Lebesgue measure $0$
(this set corresponds to the boundaries of clusters).
For each mode $m_j\in \cM$, we define its basin of attraction as
\begin{equation}
\cA_j = \{x\in \K: \dest(x)=m_j\}.
\end{equation}
The regions $\cA_1,\cdots, \cA_k$ are the clusters generated by mode clustering.

Now we recall three facts about an upper density level set 
$L=\{x:\ p(x) \geq\lambda\}$ (c.f. Figure~\ref{fig::vis_ms} 
left and middle panels):
\begin{itemize}
\item[1.] $L$ can be decomposed into connected components
$L = \bigcup_{\ell=1}^K \cC_\ell$,
where the $\cC_\ell$ are disjoint, connected compact sets
under regularity conditions.

\item[2.] Each $\cC_\ell$ contains at least one local mode.

\item[3.] If $\cC_\ell$ contains $s$ local modes, then
$\displaystyle \cC_\ell = \cC_{\ell,1}^{\lambda}\cup \cdots\cup \cC_{\ell,s}^{\lambda}\cup \cB$,
where $\cC_{\ell, j}^{\lambda}$ is the basin of attraction for a local mode $m_{\ell, j}$ intersected
with the level set $L$ and $\cB$ are the boundaries of the basins which has $0$ Lebesgue measure.
Namely, $\cC_{\ell, j}^{\lambda} = L\cap \cA_k$ for some $k$.
\end{itemize}
Thus, the upper level sets are covered by the basins of attraction of the
local modes (the uncovered regions have Lebesgue measure $0$ so we ignore them
for visualization). 

We then create a graph $G=(V,E)$ with each node corresponding to
a local mode within $L_h$ and each edge representing a \emph{connection} 
between local modes to represent a level set.
Specifically, we add an edge to a pair of local modes $(m_{\ell,j},m_{\ell,k})$
when the corresponding basins $\cC_{\ell,j}^{\lambda}$ and
$\cC_{\ell,k}^{\lambda}$ shares the same boundaries. i.e.
$\bar{\cC}_{\ell,j}^{\lambda}\cap\bar{\cC}_{\ell,k}^{\lambda}\neq \phi$.
Two local modes have an edge only if 
they are in the same connected component
and their shared boundaries are also in the upper level set.
Figure~\ref{fig::vis_ms} provides an example illustrating these parts.

\begin{figure}
\centering
\includegraphics[width=2.5in]{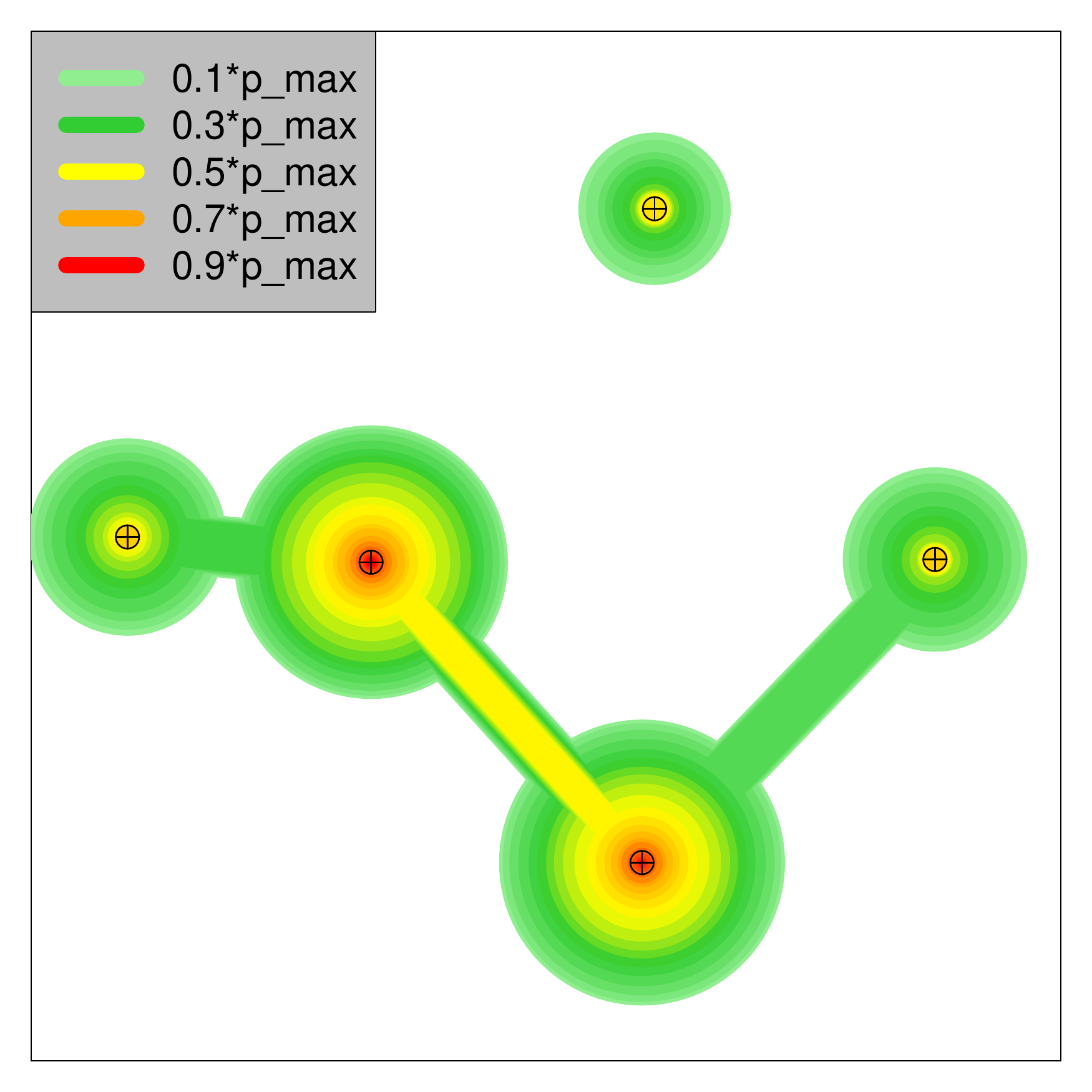}
\includegraphics[width=2.5in]{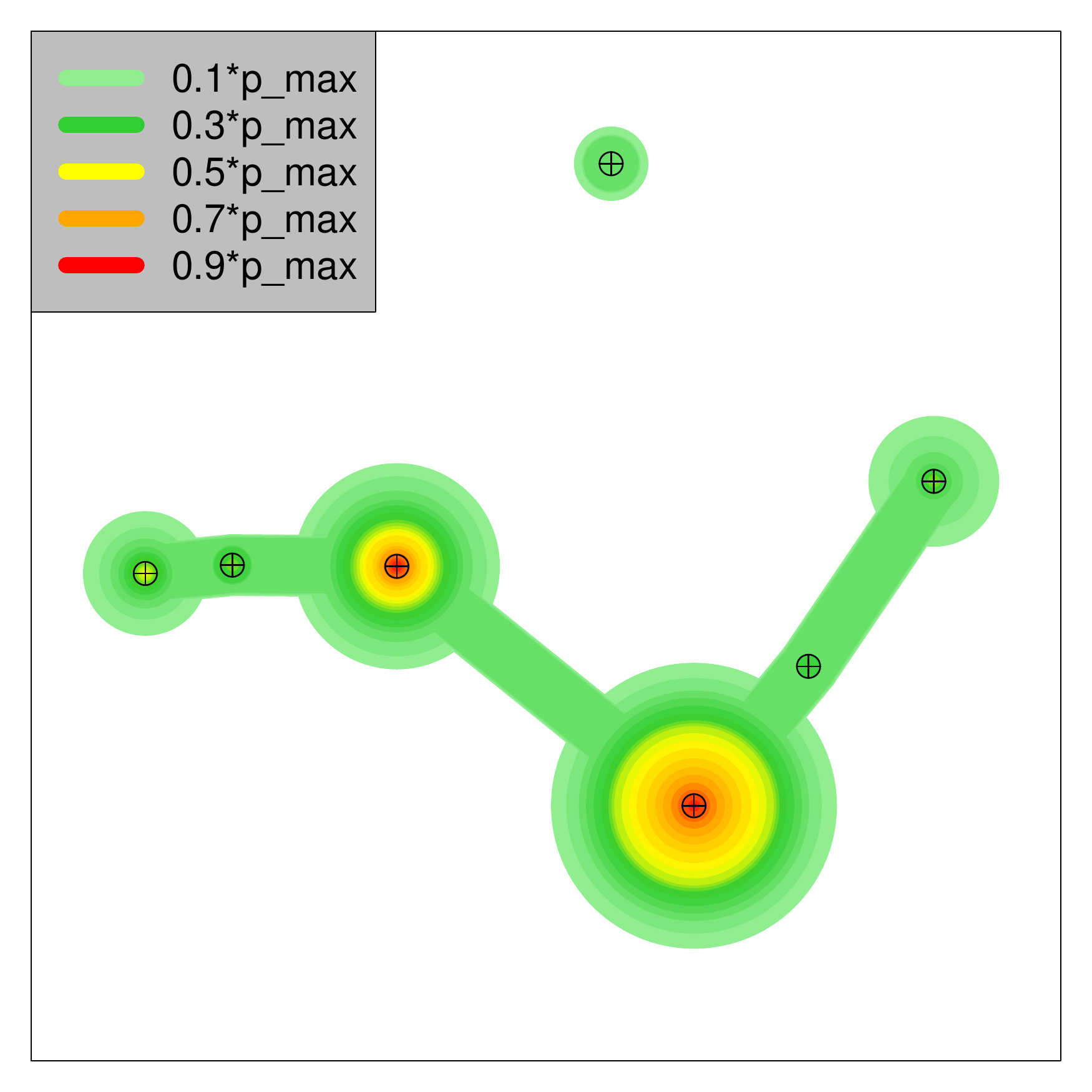}
\caption{A visualization example in multivariate cases.
We use the same dataset as Figure~\ref{fig::vis0} 
but add an additional Gaussian noise to each data point
to make it a higher dimensional dataset.
Left: a 6 dimensional dataset (add 3 additional dimensional noises). 
Right: a 10 dimensional dataset (add 7 additional dimensional noises). }
\label{fig::vis1}
\end{figure}

\subsection{Visualization Algorithm}

We construct visualization using Algorithm \ref{Alg::vis}.
First, we perform mode clustering 
to obtain local modes $m_1,\cdots, m_k\in\R^d$. 
We use the mean shift algorithm
\citep{fukunaga1975estimation,cheng1995mean, comaniciu2002mean} 
to find the modes of density estimate.
Second, we perform multidimensional scaling 
(MDS; \citealt{kruskal1964multidimensional})
to all the modes to map them onto a 2D plane.
Let $m^\dagger_1,\cdots, m^\dagger_k\in\R^2$ be the corresponding locations after MDS.
Third, we assign local mode $m_\ell$ an index
\begin{equation}
r_\ell(\lambda) = \frac{n_\ell(\lambda)}{n},\quad n_\ell(\lambda) 
= \sum_{i=1}^n 1\left(\hat{\dest}(X_i)=m_\ell\wedge \hat{p}_h(X_i)\geq \lambda\right),
\label{eq::m_idx}
\end{equation}
where $n_\ell(\lambda)$ is the number of data points that are assigned to
$m_\ell$ by mode clustering and have (estimated) density being greater or equal to $\lambda$.
Fourth,
if $r_\ell(\lambda)>0$, 
we create a circle around each $m^\dagger_\ell$
with radius proportional to $r_\ell(\lambda)$
If $r_\ell(\lambda)=0$, we ignore this local mode;
this occurs if the (estimated) density of the mode (and its basin of attractions)
lies below $\lambda$.
Fifth,
we connect two local modes $m^\dagger_\ell$ and $m^\dagger_j$
if they belong to the same connected component of
the estimated upper level set $\hat{L}_h =\{x: \hat{p}_h(x)\geq \lambda\}$
and their basins of attraction intersect $\hat{L}_h$ share the same boundary 
(see e.g. Figure~\ref{fig::vis_ms} middle and right panel).
Note that this step might be computationally difficult in practice.
An efficient alternative is to examine the shortest distance between two connected
components; if the distance is sufficiently small, 
we claim these two components share the same boundary.
Finally, 
we set the width of the line connecting $m_\ell$ and $m_j$
to be proportional to $r_\ell(\lambda)+r_j(\lambda)$.

Given several density levels, $\lambda_1<\cdots<\lambda_K$,
we can overlay the visualization from the previous paragraph from $\lambda_1$ to
$\lambda_K$ to create a ``tomographic'' visualization of the clusters.
This gives a visualization for the density level sets.
Figure~\ref{fig::vis1} shows an example for visualizing level sets for a $6$-dimensional 
and a $10$-dimensional
simulation datasets at different density levels. 
This dataset is from \cite{chen2014enhanced}.

\begin{algorithm}
\caption{Visualization for a single level set}
\label{Alg::vis}
\begin{algorithmic}
\State \textbf{Input:} Data $\{ X_1,...X_n\}$, density level $\lambda$, smoothing parameter $h$
\State 1. Compute the kernel density estimator $\hat{p}_h$.
\State 2. Find the modes $m_1,\cdots, m_k$ of $\hat{p}_h$ 
(one can apply the mean shift algorithm).
\State 3. Apply multidimensional scaling to the modes to project them into $\mathbb{R}^2$;
denote the corresponding locations as $m^\dagger_1,\cdots, m^\dagger_k$.
\State 4. For each local modes $m_\ell$, we assign it an index \eqref{eq::m_idx}
$$
r_\ell(\lambda) = \frac{n_\ell(\lambda)}{n},\quad n_\ell(\lambda) 
= \sum_{i=1}^n 1\left(\hat{\dest}(X_i)=m_\ell\wedge \hat{p}_h(X_i)\geq \lambda\right).
$$
\State 5. If $r_\ell(\lambda)>0$, we create a circle around $m^\dagger_\ell$
with radius in proportional to $r_\ell(\lambda)$.
\State 6. We connect two local modes $m^\dagger_\ell$ and $m^\dagger_j$
if
\begin{itemize} 
\item[(1)] they belong to the same connected component of $\hat{L}_h$, the estimated upper level
set at level $\lambda$,
and 
\item[(2)] their basins of attraction above level set share the same boundary.
\end{itemize}
\State 7. Adjust the width for the line connecting $m^\dagger_\ell$ and $m^\dagger_j$
to be in proportion to $r_\ell(\lambda)+r_j(\lambda)$.
\end{algorithmic}
\end{algorithm}

\subsection{Visualizing Level Sets with Confidence}	\label{sec::vis::CI}
\begin{figure}
\centering
\includegraphics[width=2.5in]{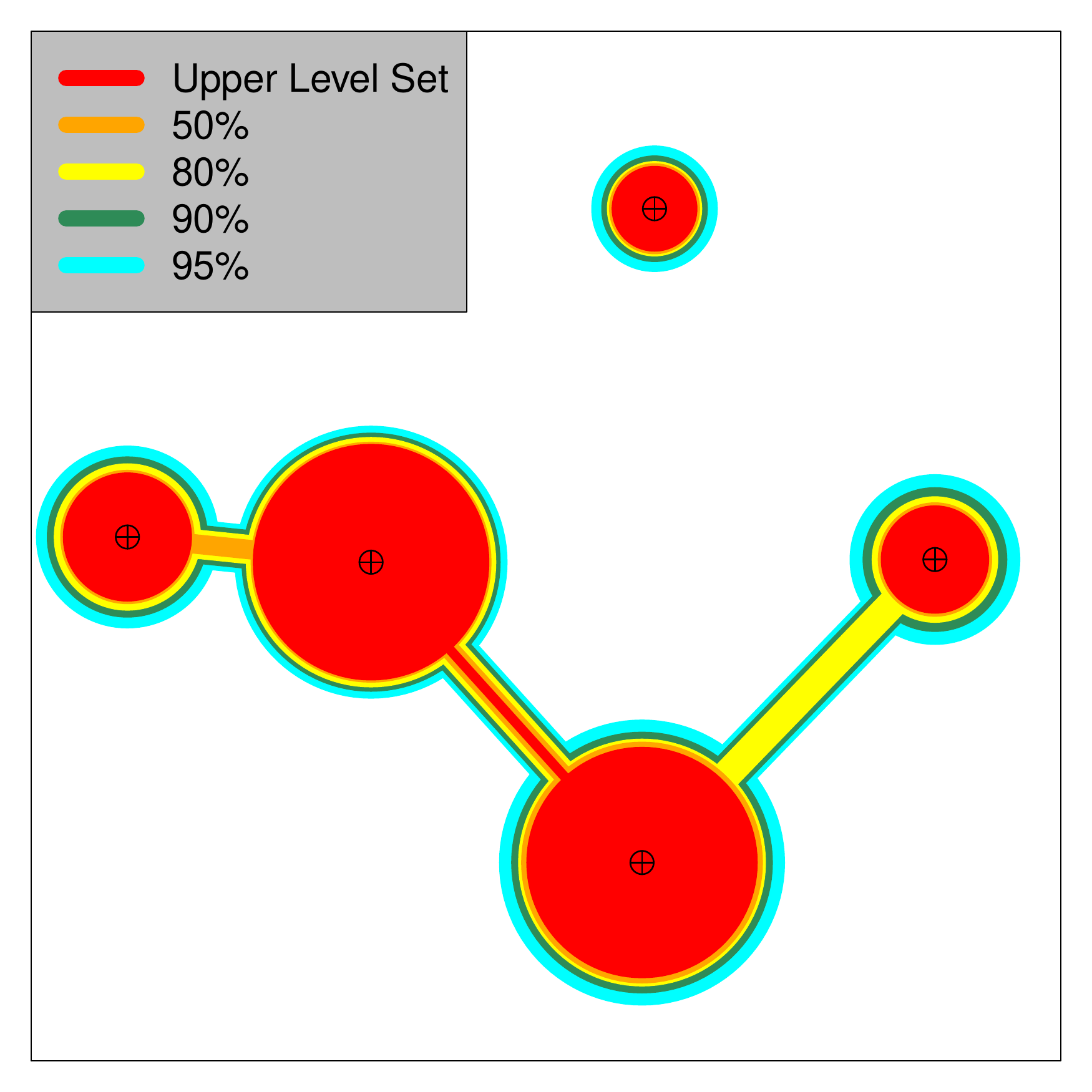}
\caption{A visualization for confidence sets of an upper set with level $\lambda$.
This is the 6 dimensional dataset of Figure~\ref{fig::vis1}.
We pick the density level $\lambda= 0.4\times\sup_x\norm{\hat{p}_h(x)}$
and display confidence levels under $\alpha=(50\%, 80\%, 90\%, 95\%)$.}
\label{fig::vis2}
\end{figure}

A modified version of our algorithm 
allows us to visualize multivariate confidence sets for
an upper level set at a given level $\lambda$.
In particular, we visualize the confidence set for the level sets
produced by Method 2 (supremum loss, see Section \ref{sec::CI::M2}) .
The advantage of Method 2 here is 
that this confidence set under different $\alpha$'s
is just the density level set at different $\lambda$'s.
This makes it easy to visualize $K$ distinct confidence sets for pre-specified levels $\alpha_1<\cdots<\alpha_K$.

Recall that $\hat{M}^*_n = \sup_{x\in\K}|\hat{p}_h^*(x)-\hat{p}_h(x)|$
is the bootstrap supremum deviation for the KDE
and $\hat{m}_{1-\alpha} = F^{-1}_{M^*_n}(1-\alpha)$ is the quantile.
When we want to visualize confidence sets for the upper level set $L_h$
at different significance levels, we pick
\begin{equation}
\lambda_1 = \lambda-\hat{m}_{1-\alpha_1},\cdots, \lambda_K = \lambda-\hat{m}_{1-\alpha_K}
\end{equation}
and use the visualization algorithm to create a tomographic visualization.
Figure~\ref{fig::vis2} provides an example for visualizing confidence
sets for upper level set using the 6 dimensional simulation data in Figure~\ref{fig::vis1}.


\begin{remark}
At the cost of additional computation,
we can also use Method 1 (Hausdorff loss) instead of Method 2,
combining it with the basins of attractions to visualize
the confidence sets.
We use method 1 to construct the inner/outer confidence sets
and then we find the connected components and partition it
by the basins of attraction for local modes
and use multidimensional scaling to visualize it in low dimensions.
\end{remark}


\section{Discussion}	
\label{sec::discuss}
In this paper, we derived the limiting distribution
for smoothed density level sets under Hausdorff loss.
This result immediately allows us to
construct confidence sets for the smoothed level set.
We developed two bootstrapping methods to construct the confidence sets,
and we showed that both methods are consistent.
These confidence sets can be inverted to construct multiple local tests
of whether a point's density value is above or below a given level,
which has application to related problems such as anomaly detection.
Finally, we developed a new visualization method 
that is informative and interpretable,
even in multidimensions.

Although we focused on density level sets in this paper,
our methods,
including confidence sets and visualization,
can be applied to a more general class of problems
such as kernel classifier, two-sample tests, and generalized level sets
(See \citealt{mammen2013confidence} for more details.).

\bibliographystyle{abbrvnat}
\bibliography{LV.bib}

\section{Proofs}
\setcounter{thm}{7}
\label{section::proofs}

\begin{thm}[Theorem 2 in \cite{Cuevas2006}]
Assume (K1--2) and (G), then we have
$$
\Haus\left(\hat{D}_h, D_h\right) = O(\norm{\hat{p}_h-p_h}_{0,\max}).
$$

\label{thm::rate}
\end{thm}

\begin{thm}[Talagrand's inequality; version of Theorem 12 in \cite{chen2014nonparametric}]
Assume (K1--2), then for each $t>0$ there exists some $n_0$ such that whenever
$n>n_0$, we have
$$
\mathbb{P}\left(\norm{\hat{p}_h-p_h}^*_{\ell, \max}>t\right)\leq (\ell+1)e^{-tnh^{d+2\ell}A_1},
$$
for some constant $A_1$ and $\ell=0,1,2$.
Moreover, 
$$
\mathbb{E}\left(\norm{\hat{p}_h-p_h}^*_{2,\max}\right) = O\left(\sqrt{\frac{\log n}{nh^{d+4}}}\right).
$$

\label{thm::tala}
\end{thm}

\begin{proof}[ for Lemma 1]
We first prove the lower bound for $\reach(D_h)$ and then
we will prove the additional assertions.

\vspace{0.2 in}
{\bf Part 1: Lower bound on reach.}
We prove this by contradiction.
Take $x$ near $D_h$ such that 
\begin{equation}
d(x, D_h) < \left(\frac{\delta_0}{2}, \frac{g_0}{\norm{p_h}^*_{2,\max}}\right).
\end{equation}
We assume that $x$ has two projections onto $D_h$, denoted as $b$ and $c$.

Since $b,c\in D_h$, $p_h(b)-\lambda = p_h(c)-\lambda= 0$
so that $p_h(b)-p_h(c)=0$.
Now by Taylor's theorem
\begin{equation}
\begin{aligned}
\norm{(b-c)^T \nabla p_h(b)} 
&= \norm{p_h(b)-p_h(c) - (b-c)^T \nabla p_h(b)}\\
& \leq \frac{1}{2}\norm{b-c}^2 \norm{p_h}_{2,\max}.
\end{aligned}
\label{eq::normal1}
\end{equation}

By the definition of projection, we can find a constant $t_b\in \R$ such that 
$x-b =t_b \nabla p_h(b)$. Together with \eqref{eq::normal1}, 
\begin{equation}
\begin{aligned}
2|(b-c)^T(x-b)| &= 2|(b-c)^T\nabla p_h(b) t_b|\\
&\leq \norm{(b-c)^T \nabla p_h(b)} |t_b|\\
&\leq \norm{p_h}_{2,\max} \norm{b-c}^2 |t_b|.
\end{aligned}
\end{equation}
Since both $b$ and $c$ are projection points from $x$ onto $D_h$, 
$$
\norm{x-b} = \norm{x-c}.
$$
Thus, we have 
\begin{equation}
\begin{aligned}
0&= \norm{x-c}^2 - \norm{x-b}^2\\
& = \norm{b-c}^2 +2(b-c)^T(x-b)\\
&\geq \norm{b-c}^2 -\norm{p_h}_{2,\max} \norm{b-c}^2 |t_b|\\
& = \norm{b-c}^2 (1-\norm{p_h}_{2,\max}|t_b|).
\end{aligned}
\label{eq::normal2}
\end{equation}

Recall that $d(x,D_h)\leq \frac{g_0}{\norm{p_h}_{2,\max}}$ and by Taylor's theorem,
\begin{equation}
\frac{g_0}{\norm{p_h}_{2,\max}} >
d(x, D_h) = \norm{x-b}=\norm{t_b\nabla p_h(b)}  = |t_b| \norm{\nabla p_h(b)}
\geq  |t_b|g_0
\end{equation}
so that $|t_b|\norm{p_h}_{2,\max} < 1$.
Note that the lower bound $g_0$ in the last inequality is because $d(x,D_h)< \frac{\delta_0}{2}$
so it follows from assumption (G).
Plugging in this result into the last equality of \eqref{eq::normal2},
we conclude that $\norm{b-c}=0$.
This shows $b=c$ so that we have a unique projection.
Thus, whenever $d(x, D_h)< \left(\frac{\delta_0}{2}, \frac{g_0}{\norm{p_h}^*_{2,\max}}\right)$,
we have a unique projection onto $D_h$
and thus we have proved the lower bound on reach.

\vspace{0.2 in}
{\bf Part 2: The three assertions.}
The first assertion is trivially true when $\norm{p_h-q}^*_{2,\max}$
is sufficiently small since assumption (G) only involves gradients (first derivatives).

The second assertion follows from the lower bound on reach.
By assertion 1, (G) holds for $q$. And the lower bound on reach 
is bounded by gradient and second derivatives so that 
we have the prescribed bound.

The third assertion follows from Theorem 1 in \cite{Chazal2007}
which states that if two $d-1$ dimensional smooth manifolds $M_1$ and $M_2$
have Hausdorff distance being less than $(2-\sqrt{2}) \min\{\reach(M_1), \reach(M_2)\}$,
then $M_1$ and $M_2$ are normal compatible to each other.
Now by Theorem \ref{thm::rate}, the Hausdorff distance between $D_h$ and $D(q)$
is at rate $O(\norm{p_h-q}_{1,\max})$
so that this assertion is true when $\norm{p_h-q}_{2,\max}$ is sufficiently small.
\end{proof}

\begin{proof}[ of Lemma~2]
Let $x\in D_h$.
We define $\Pi(x)\in D_h$ to be the projected point onto $\hat{D}_h$.
By Lemma~1 and Theorem~\ref{thm::rate},
when $\norm{\hat{p}_h-p_h}^*_{2,\max}\rightarrow 0$,
$ \Haus(D_h, \hat{D}_h) \overset{P}{\rightarrow} 0$ so that $\Pi(x)$ is unique.
Thus, we assume $\Pi(x)$ is unique.

Now since $\Pi(x)\in \hat{D}_h$ and $x\in D_h$, $\hat{p}_h(\Pi(x))-p_h(x)=0$.
Thus, by Taylor's theorem
\begin{equation}
\begin{aligned}
\hat{p}_h(x)-p_h(x) &=\hat{p}_h(x)- \hat{p}_h(\Pi(x))\\
&=  (x-\Pi(x))^T(\nabla\hat{p}_h(\Pi(x)) +O_\P(\norm{x-\Pi(x)})).
\label{eq::emp1}
\end{aligned}
\end{equation}
Note that $x-\Pi(x)$ is normal to $\hat{D}_h$ at $\Pi(x)$
so that it points toward the same direction as $\nabla \hat{p}_h(\Pi(x))$.
Thus, \eqref{eq::emp1} can be rewritten as
\begin{equation}
\hat{p}_h(x)-p_h(x)  = \norm{x-\Pi(x)}\big( \norm{\nabla \hat{p}_h(\Pi(x))} +O_\P(\norm{x-\Pi(x)})\big).
\label{eq::emp2}
\end{equation}

By Taylor's theorem,
$\nabla\hat{p}_h(\Pi(x))$ is close to $\nabla p_h(x)$
in the sense that
\begin{equation}
\nabla\hat{p}_h(\Pi(x)) = \nabla p_h(x) + O(\norm{\hat{p}_h-p_h}^*_{1,\max}).
\label{eq::emp3}
\end{equation}
In addition, $O(\norm{x-\Pi(x)}))$ is bounded 
by $O(\Haus(\hat{D}_h,D_h))$ which is at rate $O(\norm{\hat{p}_h-p_h}^*_{1,\max})$
due to Theorem~\ref{thm::rate}.
Putting this together with \eqref{eq::emp2}, we conclude
\begin{equation}
\begin{aligned}
\hat{p}_h(x)-p_h(x)
& = \norm{x-\Pi(x)} \big(\norm{p_h(x)} + O(\norm{\hat{p}_h-p_h}^*_{1,\max})\big)\\
& = d(x, \hat{D}_h)\big(\norm{p_h(x)} + O(\norm{\hat{p}_h-p_h}^*_{1,\max})\big).
\end{aligned}
\label{eq::emp4}
\end{equation}

Note that the left hand side can be written as
\begin{equation}
\hat{p}_h(x) - p_h(x) = \frac{1}{nh^d}\sum_{i=1}^n K\left(\frac{x-X_i}{h}\right) - 
\frac{1}{h^d}\mathbb{E}\left(K\left(\frac{x-X_i}{h}\right)\right) = \frac{1}{\sqrt{n}h^d}\mathbb{G}_n(\tilde{f}_x),
\label{eq::emp5}
\end{equation}
where $\tilde{f}_x(y) = K\left(\frac{x-y}{h}\right)$.
After plugging \eqref{eq::emp5} into the left hand side of \eqref{eq::emp4}, dividing both side by
$\norm{p_h(x)}$ and setting $f_x(y) = \frac{\tilde{f}_x(y)}{\sqrt{h^d}\norm{p_h(x)}}$, we
obtain
\begin{equation}
\frac{\frac{1}{\sqrt{nh^d}}\mathbb{G}_n(f_x)-d(x, \hat{D}_h)}{d(x, \hat{D}_h)} = 
O(\norm{\hat{p}_h-p_h}^*_{1,\max}).
\end{equation}
This holds uniformly for all $x\in D_h$
and note that the definition of $\cF$ is
$$
\cF = \left\{f_x(y)\equiv\frac{1}{\sqrt{h^d}\norm{ \nabla p_h(x)}}K\left(\frac{x-y}{h}\right): x\in D_h\right\}.
$$
So we conclude
$$
\sup_{x\in D_h}\left|\frac{\frac{1}{\sqrt{nh^d}}\mathbb{G}_n(f_x)-d(x, \hat{D}_h)}{d(x, \hat{D}_h)} \right|
= O(\norm{\hat{p}_h-p_h}^*_{1,\max}).
$$
\end{proof}

\begin{proof}[ for Theorem~3]
The proof for Theorem 3 follows the same procedure 
as the proof of Theorem 6 in \cite{chen2014asymptotic}.
The proof contains two parts: Gaussian approximation
and anti-concentration.

\vspace{0.2 in}
{\bf Part 1: Gaussian approximation.}
Basically, we will show that 
$$
\sqrt{nh^d}\Haus(\hat{D}_h, D_h)\approx
\sup_{f\in \cF} |\mathbb{G}_n(f)| \approx
\sup_{f\in \cF} |\mathbb{B}(f)|,
$$
where $\mathbb{B}$ is a Gaussian process defined in (13)
of the original paper.

First, when $\norm{\hat{p}_h-p_h}$ is sufficiently small,
$\hat{D}_h$ and $D_h$ are normal compatible to each other by Lemma~1.
Then by the property of normal compatible,
\begin{equation}
\sup_{x\in D_h} d(x,\hat{D}_h) = \Haus(\hat{D}_h, D_h).
\label{eq::G1}
\end{equation}
Thus, the difference
\begin{equation}
\begin{aligned}
\left|\sqrt{nh^d}\Haus(\hat{D}_h, D_h)-
\sup_{f\in \cF} |\mathbb{G}_n(f)| \right|
&= \left|\sqrt{nh^d}\sup_{x\in D_h}d(x,\hat{D}_h)-
\sup_{f\in \cF} |\mathbb{G}_n(f)| \right|\\
&\leq \frac{\sup_{x\in D_h} \left|\frac{1}{\sqrt{nh^d}}\mathbb{G}_n(f_x)-d(x, \hat{D}_h)\right|}{\frac{1}{\sqrt{nh^d}}}\\
&= \sup_{x\in D_h}\left|\frac{\frac{1}{\sqrt{nh^d}}\mathbb{G}_n(f_x)-d(x, \hat{D}_h)}{d(x, \hat{D}_h)} \right| 
O_\P(1)\\
&=O(\norm{\hat{p}_h-p_h}^*_{1,\max}).
\end{aligned}
\label{eq::G2}
\end{equation}
Note that the last two inequality follows from the fact that $d(x,\hat{D}_h) \leq O_\P(\frac{1}{\sqrt{nh^d}})$.
By Theorem~\ref{thm::tala} the above result implies,
\begin{equation}
\mathbb{P}\left(\left|\sqrt{nh^d}\Haus(\hat{D}_h, D_h)-
\sup_{f\in \cF} |\mathbb{G}_n(f)| \right|
> t\right)\leq 2e^{-tnh^{d+2} A_2}
\label{eq::G3}
\end{equation}
for some constant $A_2$.

Now by Corollary 2.2 in \cite{chernozhukov2014gaussian}, 
there exists some random variable $\mathbf{B} \overset{d}{=} \sup_{f\in\cF}|\mathbb{B}(f)|$
such that for all $\gamma \in (0,1)$ and $n$ is sufficiently large,
\begin{equation}
\mathbb{P}\left(\left|\sup_{f\in \cF} |\mathbb{G}_n(f)| -\mathbf{B} \right|> 
A_3 \frac{\log^{2/3}(n)}{\gamma^{1/3}(nh^{d})^{1/6}}\right) \leq A_4 \gamma.
\label{eq::G4}
\end{equation}
Note that this result basically follows from the same derivation of
Proposition 3.1 in \cite{chernozhukov2014gaussian}
with the fact that $g\equiv 1$ in their definition.

Combining equations \eqref{eq::G3} and \eqref{eq::G4} and pick $t= 1/\sqrt{nh^{d+2}}$, 
we have that
for $n$ is sufficiently large and $\gamma \in (0,1)$,
\begin{equation}
\mathbb{P}\left(\left|\sqrt{nh^d}\Haus(\hat{D}_h, D_h)-\mathbf{B} \right|> 
A_3 \frac{\log^{2/3}(n)}{\gamma^{1/3}(nh^{d})^{1/6}} +\frac{1}{\sqrt{nh^{d+2}}}\right) \leq A_4 \gamma
+ 2e^{-\sqrt{nh^{d+2}} A_2}.
\label{eq::G5}
\end{equation}

\vspace{0.2 in}
{\bf Part 2: Anti-concentration.}
To obtain the desired Berry-Esseen bound, we apply the anti-concentration inequality
in \cite{chernozhukov2014gaussian} and \cite{chernozhukov2014anti}.

\begin{lem} [Modification of Lemma 2.3 in \cite{chernozhukov2014gaussian}]
\label{lem::anti}
Let $\mathbf{B} \overset{d}{=} \sup_{f\in\cF}|\mathbb{B}(f)|$,
where $\mathbb{B}$ and $\mathcal{F}$ are defined as the above.
Assume (K1-2) and that there exists a random variable $Y$ such that 
$\mathbb{P}(\left|Y-\mathbf{B}\right|>\eta)<\delta(\eta)$. 
Then
$$
\sup_{t} \left|\P(Y<t) - \P\left(\mathbf{B}<t\right) \right| \leq A_5
\mathbb{E}(\mathbf{B})\eta+\delta(\eta)
$$
for some constant $A_5$.
\end{lem}

It is easy to verify that assumptions (K1-2) imply the assumptions (A1-3)
in \cite{chernozhukov2014gaussian} so that the result follows.
Note that in the original Lemma 2.3 in \cite{chernozhukov2014gaussian},
$\mathbb{E}(\mathbf{B})$ should be replaced by $\mathbb{E}(\mathbf{B})+ \log \eta$.
However, $\mathbb{E}(\mathbf{B})=O(\sqrt{\log n})$
due to Dudley's inequality for Gaussian process \citep{van1996weak} 
and later we will find that $\log \eta$ is also at this rate so we ignore $\log \eta$.

From Lemma \ref{lem::anti} and equation \eqref{eq::G5},
there exists some constant $A_6$ such that
\begin{equation}
\begin{aligned}
\sup_t\Bigg|\mathbb{P}\Big(\sqrt{nh^{d}}&\Haus(\hat{D}_h, D_h)<t\Big)-
\mathbb{P}\left(\sup_{f\in\cF}|\mathbb{B}(f)|<t\right)
\Bigg|\\
&\leq A_5 
\mathbb{E}(\mathbf{B})\left(A_3 \frac{\log^{2/3}(n)}{\gamma^{1/3}(nh^{d})^{1/6}} +\frac{1}{\sqrt{nh^{d+2}}}\right)
+A_4 \gamma+ 2e^{-\sqrt{nh^{d+2}} A_2}\\
&\leq A_6 \left(A_3 \frac{\log^{7/6}(n)}{\gamma^{1/3}(nh^{d})^{1/6}} +\sqrt{\frac{\log n}{nh^{d+2}}}\right)
+A_4 \gamma+ 2e^{-\sqrt{nh^{d+2}} A_2}.
\end{aligned}
\end{equation}
Now pick $\gamma = \left(\frac{\log^7 n}{nh^{d}}\right)^{1/8}$
and use the fact that $\frac{1}{\sqrt{nh^{d+2}}}$
and $2e^{-\sqrt{nh^{d+2}} A_2}$ converges faster than the other terms;
we obtain the desired rate.
\end{proof}

\begin{proof}[ for Theorem~4]
This proof follows the same strategy for the proof of Theorem 7 in \cite{chen2014asymptotic}.
We prove the Berry-Esseen type bound first
and then show that the coverage is consistent.
We prove the Berry-Esseen bound in two simple steps:
Gaussian approximation and
support approximation.

Let $\mathcal{X}_n = \{(X_1,\cdots, X_n): \norm{\hat{p}_h-p_h}^*_{2,\max}\leq \eta_0\}$
for some small $\eta_0$
so that 
whenever our data is within $\mathcal{X}_n$,
(G) holds for $\hat{p}_h$.
By Lemma 1, such an $\eta_0$ exists
and by Theorem~\ref{thm::tala} we have $\mathbb{P}(\mathcal{X}_n)\geq 1-3^{-nh^{d+4}\tilde{A}_0}$
for some constant $\tilde{A}_0$.
Thus, we assume our original data $X_1,\cdots, X_n$ is within $\mathcal{X}_n$.

\vspace{0.2 in}
{\bf Step 1: Gaussian approximation.}
Let $\hat{\mathbb{P}}_n$ and $\hat{\mathbb{P}}^*_n$ be the empirical measure
and the bootstrap empirical measure.
A crucial observation is that for a function $f_x(y) = K\left(\frac{x-y}{h}\right)$,
\begin{equation}
\hat{\mathbb{P}}_n (f_x) = \int K\left(\frac{x-y}{h}\right)d\hat{\mathbb{P}}_n(y) = h^d\hat{p}_h(x).
\end{equation}
Also note 
\begin{equation}
\hat{\mathbb{P}}^*_n (f_x) = \int K\left(\frac{x-y}{h}\right)d\hat{\mathbb{P}}^*_n(y) = h^d\hat{p}^*_h(x).
\end{equation}
Therefore, for the bootstrap empirical process 
$\mathbb{G}_n^*= \sqrt{n}(\hat{\mathbb{P}}^*-\hat{\mathbb{P}})$,
\begin{equation}
\hat{p}^*_h(x) - \hat{p}_h(x) = \frac{1}{\sqrt{n}h^d} \mathbb{G}_n^*(f_x).
\end{equation}

Thus, if we sample from $\hat{p}_h$ and consider 
estimating $\hat{p}_h$ by $\hat{p}^*_h$,
we are doing exactly the same procedure of estimating $p_h$ by $\hat{p}_h$.
Therefore, Lemma~2 and Theorem~3
hold for approximating $\Haus(\hat{D}^*_h, \hat{D}_h)$
by a maxima for a Gaussian process.
The difference is that the Gaussian process is defined on
\begin{equation}
\cF_n = 
\left\{f_x(y)\equiv\frac{1}{\sqrt{nh^d}\norm{ \nabla \hat{p}_h(x)}}K\left(\frac{x-y}{h}\right): x\in \hat{D}_h\right\}
\end{equation}
since the ``parameter (level sets)'' being estimated is $\hat{D}_h$ (the estimator is $\hat{D}^*_h$).
Note that $\cF_n$ is very similar to $\cF$ except the denominator is slightly different and
the support $\hat{D}_h$ is also different from $D_h$.
That is, we have 
\begin{equation}
\begin{aligned}
\sup_t\Bigg|\mathbb{P}\Bigg(\sqrt{nh^d}\Haus(&\hat{D}^*_h, \hat{D}_h)<t\Bigg| X_1,\cdots,X_n\Bigg) \\
&- 
\mathbb{P}\left(\sup_{f\in\cF_n}|\mathbb{B}_n(f)|<t\Bigg| X_1,\cdots,X_n\right)\Bigg| 
\leq O\left(\left(\frac{\log^7 n}{nh^d}\right)^{1/8}\right),
\end{aligned}
\label{eq::BT1}
\end{equation}
where $\mathbb{B}_n$ is a Gaussian process on $\cF_n$ such that 
for any $f_1, f_2 \in \cF_n$,
\begin{equation}
\mathbb{E}(\mathbb{B}_n(f_1)|X_1,\cdots,X_n) = 0, 
\quad \Cov(\mathbb{B}_n(f_1), \mathbb{B}(f_2)|X_1,\cdots,X_n) = \frac{1}{n}\sum_{i=1}^n f_1(X_i)f_2(X_i).
\label{eq::BG}
\end{equation}

\vspace{0.2 in}
{\bf Step 2: Support approximation.}
In this step, we will show that 
\begin{equation}
\sup_{f\in\cF_n}|\mathbb{B}_n(f)|\approx
\sup_{f\in\cF}|\mathbb{B}_n(f)|\approx
\sup_{f\in\cF}|\mathbb{B}(f)|.
\end{equation}
The first approximation can be shown by using 
the Gaussian comparison lemma 
(Theorem 2 in \cite{chernozhukov2014comparison}; also see Lemma 17 in \cite{chen2014asymptotic}).
We do the same thing as Step 3 in the proof of Theorem 8 in \cite{chen2014asymptotic}
so we omit the details.
Essentially, given any $\epsilon>0$,
we can construct a pair of balanced $\epsilon$-nets
for both $\cF$ and $\cF_n$, denoted as $\{g_1,\cdots, g_K\}$ and 
$\{g^n_1, \cdots, g^n_K\}$
so that $\max_{j}\norm{g_j-g^n_j}^*_{\max} = O(\norm{\hat{p}_h-p_h}^*_{1,\max})$.
Then this $\epsilon$-net leads to 
\begin{equation}
\begin{aligned}
\sup_t\Bigg|\mathbb{P}\Bigg(\sup_{f\in\cF_n}&|\mathbb{B}_n(f)|<t\Bigg| X_1,\cdots,X_n\Bigg) \\
&- 
\mathbb{P}\left(\sup_{f\in\cF}|\mathbb{B}_n(f)|<t\Bigg| X_1,\cdots,X_n\right)\Bigg| 
\leq O\left(\left(\norm{\hat{p}_h-p_h}^*_{1,\max}\right)^{1/3}\right).
\end{aligned}
\label{eq::Sapp}
\end{equation}
The difference between $\sup_{f\in\cF}|\mathbb{B}_n(f)|$
and $\sup_{f\in\cF}|\mathbb{B}(f)|$
is small since the these two Gaussian processes differ 
in their covariance but as $n\rightarrow \infty$,
the covariances converges at rate $1/\sqrt{n}$
so that we can neglect the difference between them.
Thus, combining \eqref{eq::BT1} and \eqref{eq::Sapp}
and the argument from previous paragraph, we conclude
\begin{equation}
\begin{aligned}
\sup_t\Bigg|\mathbb{P}\Bigg(\sqrt{nh^d}\Haus(\hat{D}^*_h, \hat{D}_h)<t&\Bigg| X_1,\cdots,X_n\Bigg) 
- 
\mathbb{P}\left(\sup_{f\in\cF}|\mathbb{B}(f)|<t\right)\Bigg| \\
&\leq O\left(\left(\frac{\log^7 n}{nh^d}\right)^{1/8}\right)+
O\left(\left(\norm{\hat{p}_h-p_h}^*_{1,\max}\right)^{1/3}\right).
\end{aligned}
\label{eq::BT2}
\end{equation}
Now comparing the above result to Theorem~3
and using the fact that the first big-O term 
dominates the second term 
(the first is of rate $-1/8$ for $n$ but
the second term is at rate $-1/6$ by Theorem~\ref{thm::tala}), 
we conclude the result for first assertion.

For the coverage, 
let $W_n = \Haus(\hat{D}_h, D_h)$ and $w_{n,1-\alpha} = F^{-1}_{W_n}(1-\alpha)$.
Since $D_h \subset \hat{D}_h \oplus  \Haus(\hat{D}_h, D_h)$,
we have
\begin{equation}
\mathbb{P}(D_h \subset \hat{D}_h \oplus w_{n,1-\alpha} ) = 1-\alpha.
\end{equation}
Now by the first assertion, the difference for
$w_{n,1-\alpha} $ and the bootstrap estimate $w^*_{n,1-\alpha} $
differs at rate $O\left(\left(\frac{\log^7 n}{nh^d}\right)^{1/8}\right)$,
which completes the proof.

\end{proof}

---
\end{document}